\documentclass[10pt,journal]{IEEEtran}
\pdfoutput=1
\usepackage{amsmath,amsthm,amsfonts,amssymb}
\newtheorem{lemma}{Lemma}

\usepackage{algorithmic}
\usepackage{algorithm}
\usepackage{array}
\usepackage[caption=false,font=normalsize,labelfont=sf,textfont=sf]{subfig}
\usepackage{textcomp}
\usepackage{stfloats}
\usepackage{url}
\usepackage{verbatim}
\usepackage{graphicx}
\usepackage{float}
\usepackage{cite}
\usepackage{diagbox}
\allowdisplaybreaks[3]

\usepackage[pdfencoding=auto,
colorlinks,
linkcolor=black,
anchorcolor=black,
urlcolor = black,
citecolor=blue]{hyperref}
\begin{document}


 \title{Integrated Sensing and Communication with Reconfigurable Distributed Antenna and Reflecting Surface: Joint Beamforming and Mode Selection}

\author{Pingping Zhang, Jintao Wang, Yulin Shao, and Shaodan Ma
\vspace{-3mm}
\thanks{The authors are with the State Key Laboratory of Internet of Things for Smart City (SKL-IoTSC) and the Department of Electrical and Computer Engineering, University of Macau, Macau S.A.R., China (e-mails:  yc17433@um.edu.mo; wang.jintao@connect.um.edu.mo; ylshao@um.edu.mo; shaodanma@um.edu.mo).}}

\maketitle
\begin{abstract}
This paper presents a new integrated sensing and communication (ISAC) framework, leveraging the recent advancements of reconfigurable distributed antenna and reflecting surface (RDARS).
RDARS is a programmable surface structure comprising numerous elements, each of which can be flexibly configured to operate either in a reflection mode, resembling a passive reconfigurable intelligent surface (RIS), or in a connected mode, functioning as a remote transmit or receive antenna.
Our RDARS-aided ISAC framework effectively mitigates the adverse impact of multiplicative fading when compared to the passive RIS-aided ISAC, and reduces cost and energy consumption when compared to the active RIS-aided ISAC.
Within our RDARS-aided ISAC framework, we consider a radar output signal-to-noise ratio (SNR) maximization problem under communication constraints to jointly optimize the active transmit beamforming matrix of the base station (BS), the reflection and mode selection matrices of RDARS, and the receive filter. To tackle the inherent non-convexity and the binary integer optimization introduced by the mode selection in 
 this optimization challenge, we propose an efficient iterative algorithm with proved convergence based on majorization minimization (MM) and penalty-based methods.  Numerical and simulation results demonstrate the superior performance of our new framework, and clearly verify substantial distribution, reflection as well as selection gains obtained by properly configuring the RDARS. 
\end{abstract}

\vspace{-2mm}
\section{Introduction}
\subsection{Background}
By 2025, it is expected that 75.4 billion devices will be interconnected within the vast landscape of the Internet of Things (IoT) \cite{1}. This rapid expansion lays the foundation for an interconnected future world, but it also places a heavy burden on wireless communication networks, intensifying the issue of spectrum congestion.
On the other hand, emerging applications introduced by IoT, such as smart homes, robot navigation, and autonomous vehicles, demand not only high-quality wireless connectivity, but also an elevated capacity to perceive and comprehend the environment\cite{zhang2021enabling}. These pose new challenges and requirements for next-generation wireless communication systems.

In response to the above challenges, integrated sensing and communication (ISAC) emerges as a promising technology \cite{liu2020joint,liu2022survey,liu2022integrated,yuan2021integrated}. 
ISAC opens the door for communication to tap into spectrum resources initially designated for radar, and utilizes the reflections of communication signals to sense the environment. This transformation acts as a catalyst, propelling us away from conventional communication systems towards a new paradigm where communication and sensing seamlessly intertwine.
Furthermore, by employing a unified hardware platform for both communication and sensing, ISAC drives remarkable enhancements in hardware efficiency\cite{8892631,wang2023stars}.

ISAC, however, faces a significant challenge in practice when line-of-sight (LoS) links are unavailable between the transmitter, e.g., a base station (BS), and the sensing target. In such scenarios, the sensing performance can significantly deteriorate.
One intriguing solution to tackle this challenge involves the deployment of reconfigurable intelligent surfaces (RIS).
In essence, an RIS is a planar metasurface composed of numerous passive elements, each capable of adjusting the phase and amplitude of incident signals, effectively reshaping the wireless propagation environment \cite{10056867,zhang2023double,9429987}. When integrated into the ISAC system, RIS opens up the possibility of creating virtual LoS links between the BS and the target, particularly in areas where physical obstacles would otherwise obstruct the signal's path.
\vspace{-2mm}
\subsection{Related work}
RIS-aided ISAC systems have garnered substantial interests in recent research endeavors. Many of these studies focus on fully-passive RIS setups \cite{9364358,luo2023ris,xing2022passive,9844707,liu2022joint,xu2023joint,he2022ris,9591331,10138058}. For example,
in the pioneering work \cite{9364358}, the authors first introduced the deployment of RIS in a dual-function radar and communication (DFRC) system.
They jointly designed the transmit beamforming matrix and RIS passive beamforming matrix by forming a radar signal-to-noise ratio (SNR) maximization problem.
Refs. \cite{luo2023ris,xing2022passive,9844707,liu2022joint,xu2023joint} delved into the active and passive beamforming design of RIS-aided ISAC systems, each with unique design criteria ranging from weighted radar SNR summation maximization, user SNR maximization, weighted combination of the radar mutual information (MI) and the communication sum-rate maximization, radar output signal-to-interference-plus-noise ratio (SINR) maximization and the radar MI maximization.


Nevertheless, it is well understood that passive RIS-aided systems introduce a unique challenge: signals must traverse both the transmitter-RIS and RIS-receiver links, giving rise to the ``multiplicative fading'' effect. In simpler terms, the path loss of the transmitter-RIS-receiver link is not the sum of the path losses of the transmitter-RIS and RIS-receiver links but rather their product\cite{9998527}. While this can be mitigated by deploying a large number of reflecting elements, it inevitably leads to heightened signal processing complexity.
To overcome these limitations, the concepts of active RIS\cite{long2021active} and hybrid RIS\cite{sankar2022beamforming}
have emerged as promising alternatives. 

Diverging from the passive RIS, the active RIS comprises numerous active elements equipped with additional power amplifiers (PAs) in the circuit. This setup allows the active RIS to not only reflect but also amplify the incident signals. Several research efforts have been dedicated to exploring the potential of active RIS-aided ISAC systems \cite{salem2022active,yu2023active,9979782,10184278}. For instance, 
the authors in \cite{yu2023active} maximized the radar output SINR by jointly optimizing the transmit beamforming matrix at the BS and the reflection matrix at the active RIS, and the simulation results revealed that the active RIS can achieve more than 60 dB radar SNR increasement in contrast to the passive RIS.

When it comes to the hybrid RIS, it combines both passive and active elements, effectively striking a balance between the passive and active RIS. 
 The authors in \cite{sankar2022beamforming} investigated the potential of the hybrid RIS in the ISAC system by formulating a  worst-case target illumination power maximization problem.




Active RIS and hybrid RIS incorporate the PAs into the circuits and allow for signal amplification prior to reflection, ultimately bolstering signal strength at the receiver's end.
As a result, they can counteract the ``multiplicative fading'' effect. However, it is important to note that these advancements come at the cost of increased power consumption\footnote{The increased power consumption is necessary for signal amplification to compensate the two concatenated transmission losses in the  transmitter-RIS and RIS-receiver links.} and the need for more intricate hardware implementations. Moreover, timely control of various RISs to adapt to the channel changes is indispensable for unleashing the RIS potential in performance enhancement. The RIS controlling issue is generally overlooked and practical solution is  yet to be investigated. 
To address these limitations and strike a balance between power efficiency and signal enhancement, our recent work introduced a new architecture known as the reconfigurable distributed antenna and reflecting surface (RDARS) \cite{ma2023reconfigurable}. The hardware platform of RDARS-aided  multi-user communication system is shown in Fig.~\ref{fig:demo}.

\begin{figure}[t]
    \centering
    \includegraphics[width=3.1in]{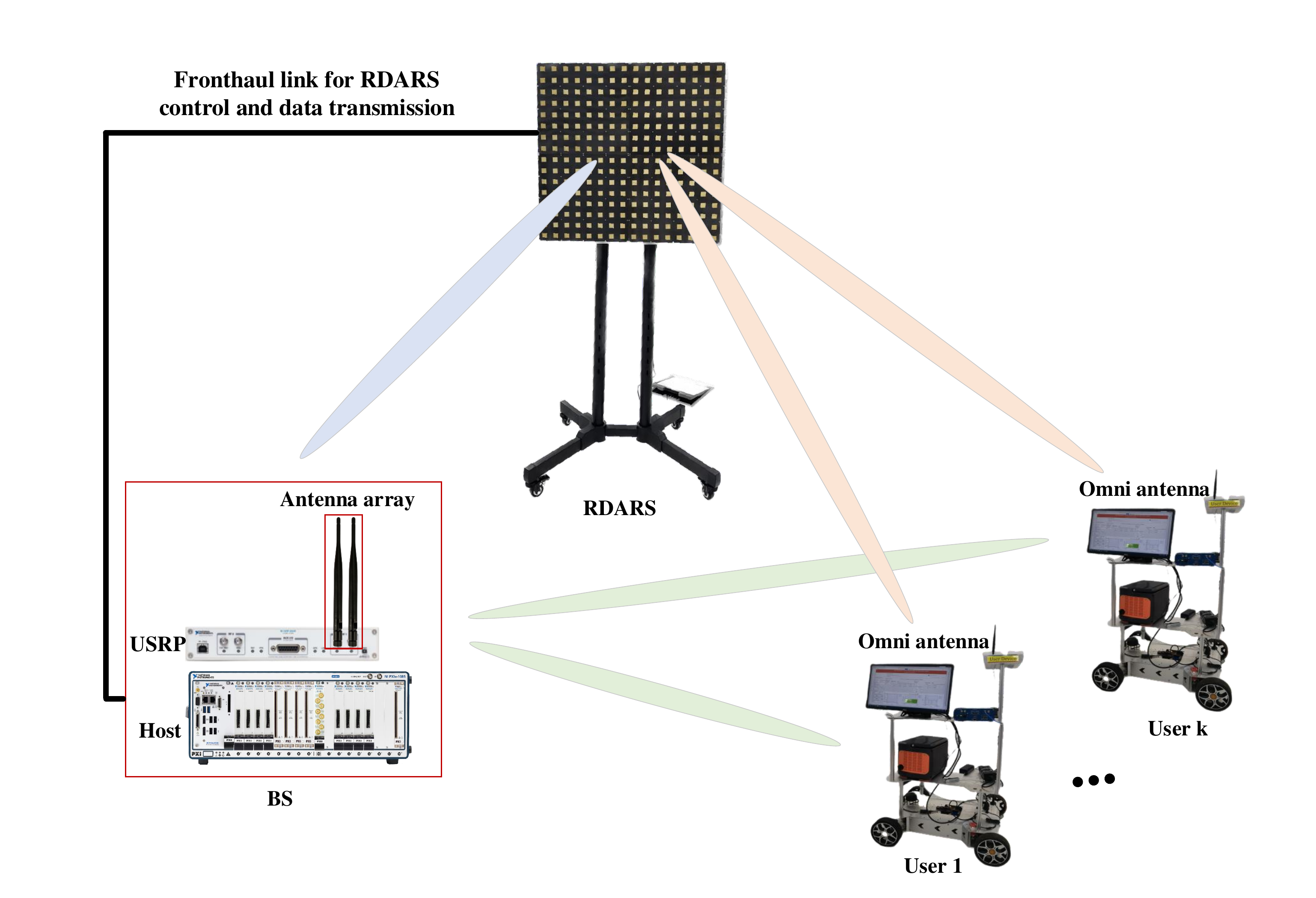}
    \vspace{-2mm}
    \caption{The hardware platform of RDARS-aided multi-user communication system.}
   \label{fig:demo}
\end{figure} 

RDARS draws inspiration from both RIS and distributed antenna systems (DAS). In this architecture, each individual element is capable of two distinct modes: the reflection mode and the connected mode.
In the reflection mode, an element functions like a traditional fully-passive RIS, whereas in the connected mode, it operates as a remote transmit or receive antenna connected to the BS through wires or fibers. This dual-mode capability bestows upon RDARS the benefits of both ``reflection gain'' from the reflecting elements and ``distribution gain'' from the connected elements. Thus, the ``multiplicative fading'' effect can be compensated without much additional power consumption and the timely RDARS control can be well solved by leveraging the fronthaul connection between the RDARS and BS. Furthermore, RDARS exhibits a remarkable degree of flexibility and generality, encompassing DAS and RIS as special cases within its framework.
The experimental results in \cite{ma2023reconfigurable} revealed that, even when one element operates in the connected mode, the RDARS-assisted communication system provides a $21\%$ gain over the DAS-assisted system, and a $170\%$ gain over the passive-RIS assisted system. These results demonstrate the great potential of RDARS for enhancing the performance of wireless communication systems.


\vspace{-2mm} 
\subsection{Contributions}
In recognition of the limitations of RIS-aided ISAC systems, the paper leverages the advantages of RDARS and introduces a new RDARS-aided ISAC framework.

Our main contributions are summarized as follows:

\begin{itemize}
\item We put forth a new RDARS-aided ISAC framework, where the BS communicates with users and senses the target simultaneously with the help of a RDARS. The RDARS comprises multiple elements, each of which can be dynamically configured to function in either the reflection mode or the connected mode.
In contrast to passive RIS-aided ISAC systems, RDARS-aided ISAC effectively mitigates the detrimental effects of multiplicative fading, thereby enhancing communication reliability. 
Compared with active RIS-aided ISAC, RDARS-aided ISAC offers advantages in terms of cost-efficiency and reduced energy consumption.
\item Within our RDARS-aided ISAC framework, we formulate a joint beamforming and mode selection problem to optimize the transmit beamforming matrix of the BS, the reflection and mode selection matrices of RDARS, and the radar receive filter. To address this non-convex radar output SNR maximization challenge under stringent communication requirements, we develop an efficient algorithm that combines alternation optimization (AO), majorization-minimization (MM) and penalty techniques. In the first step, the AO algorithm dissects the original problem into multiple tractable sub-problems, yielding optimal solutions for the transmit beamforming and receive filter. Then, by applying penalty and MM techniques, the optimal RDARS reflection matrix and the mode selection matrix can be derived in closed forms.
\item The superior performance achieved by our RDARS-aided ISAC framework is verified through numerical and simulation results. It is shown that RDARS-aided ISAC remarkably outperforms the passive RIS-aided ISAC. This superiority arises from RDARS's ability to not only provide passive reflection gains but also harness the distribution gain offered by elements operating in the connected mode. Furthermore, our proposed joint beamforming optimization algorithm optimizes the placement of elements in the connected mode, thus bringing out additional selection gain to substantially enhance the performance of the proposed RDARS-aided ISAC framework.
\end{itemize}

The remainder of this paper is organized as follows. Section \uppercase\expandafter{\romannumeral2} gives the RDARS-aided ISAC system model and the corresponding optimization problem. The proposed algorithm to solve the challenging optimization problem is presented in Section \uppercase\expandafter{\romannumeral3}. Then, we conduct extensive simulations to validate the performance of the proposed algorithm in Section \uppercase\expandafter{\romannumeral4}. Lastly, Section 
 \uppercase\expandafter{\romannumeral5} concludes this paper.

Notations: Lower-case letters, bold-face lower-case letters and bold-face upper-case letters are used for scalars, vectors and matrices, respectively. For a matrix $\mathbf{A}$, $\mathbf{A}^{T}$, $\mathbf{A}^{H}$, $\mathbf{A}^{\ast}$, $\mathrm{tr}(\mathbf{A})$, $\mathbf{A}(i,j)$, and $\lVert  \mathbf{A}\rVert_{F}$ represent its transpose, conjugate transpose, conjugate, trace, the $(i,j)$th entry and Frobenius norm respectively. $\mathrm{diag}\{\mathbf{A}\}$ represents a diagonal matrix whose diagonal elements are the same with matrix $\mathbf{A}$ while $\mathrm{diag}\{\mathbf{a}\}$ represents a diagonal matrix whose diagonal elements are the same with vector $\mathbf{a}$. $\mathbf{I}$ refers to an  identity matrix.
$\Re\{\cdot\}$ and $\Im\{\cdot\}$ denote the real part and the imaginary part of a complex input, respectively. $\lvert \cdot \rvert$ denotes the absolute value operation and $\lVert \cdot\rVert$ denotes the Euclidean norm operation. $\otimes$ stands for the Kronecker product and the angle of a complex input is denoted as $\text {arg}(\cdot)$. $\mathbb{C}^{a\times b }$ and $\mathbb{R}^{c\times d }$ denote the complex-valued space with $a\times b$ dimensions and the real-valued space with $c\times d$ dimensions, respectively. $\mathrm{E}$ denotes the expectation operation. $\mathcal {CN}({0},\sigma _{1,k} ^{2} )$ represents the distribution of a circularly symmetric complex Gaussian (CSCG) random variable with zero mean and variance $\sigma _{1,k} ^{2}$. 


\vspace{-2mm}
\section{System Model}
Consider a RDARS-assisted ISAC system, as shown in Fig. \ref{system_model}, where a BS serves $K$ users in the downlink while simultaneously sensing a point target with the assistance of a RDARS.
The BS is equipped with $M_{t}$ transmit antennas and $M_{r}$ receive antennas, both arranged in a uniform linear array (ULA) configuration. In particular, we assume $M_{t}=M_{r}=M$. Each user is equipped with a single antenna.

The RDARS is equipped with a uniform planar array (UPA) comprising a total of $N$ elements. Each element can function in either the connected mode or the reflection mode. The specific mode can be configured at the RDARS controller.
Let there be $a$ elements working in the connected mode and the remaining  $N-a$ elements working in the reflection mode. From the perspective of reducing the hardware cost, we assume that there are only a few elements working in the connected mode, i.e., $a\ll N$.
We consider dynamic RDARS, where the locations of connected elements can be configured flexibly according to different system performance requirements.
To quantify this adaptability, we introduce a diagonal matrix $\mathbf{A}=\mathrm{diag}\{a_{1},\cdots,a_{N}\}\in \mathbb{R}^{N\times N}$, named the selection matrix. Specifically, $a_{i}=\mathbf{A}(i,i)\in\{0,1\}$. $a_{i}=1$ signifies that the $i$th element operates in the connected mode, whereas $a_{i}=0$ signifies that the $i$th element operates in the reflection mode. 
Further, we define a reflection matrix $\boldsymbol{\Phi}=\mathrm{diag}\{\phi_{1},\cdots,\phi_{N}\}$, where $\phi_{i}=|\phi_{i}|e^{ j {\arg}( \phi_{i})}$ with $|\phi_{i}|=1$ and ${\arg}( \phi_{i})\in[0,2\pi)$ is the reflection coefficient of the $i$th element if it operates in the reflection mode. Thus, the equivalent reflection matrix of RDARS can be written as $(\mathbf{I}-\mathbf{A})\boldsymbol{\Phi}$.

\begin{figure}[t]
    \centering
    \includegraphics[width=3.3in]{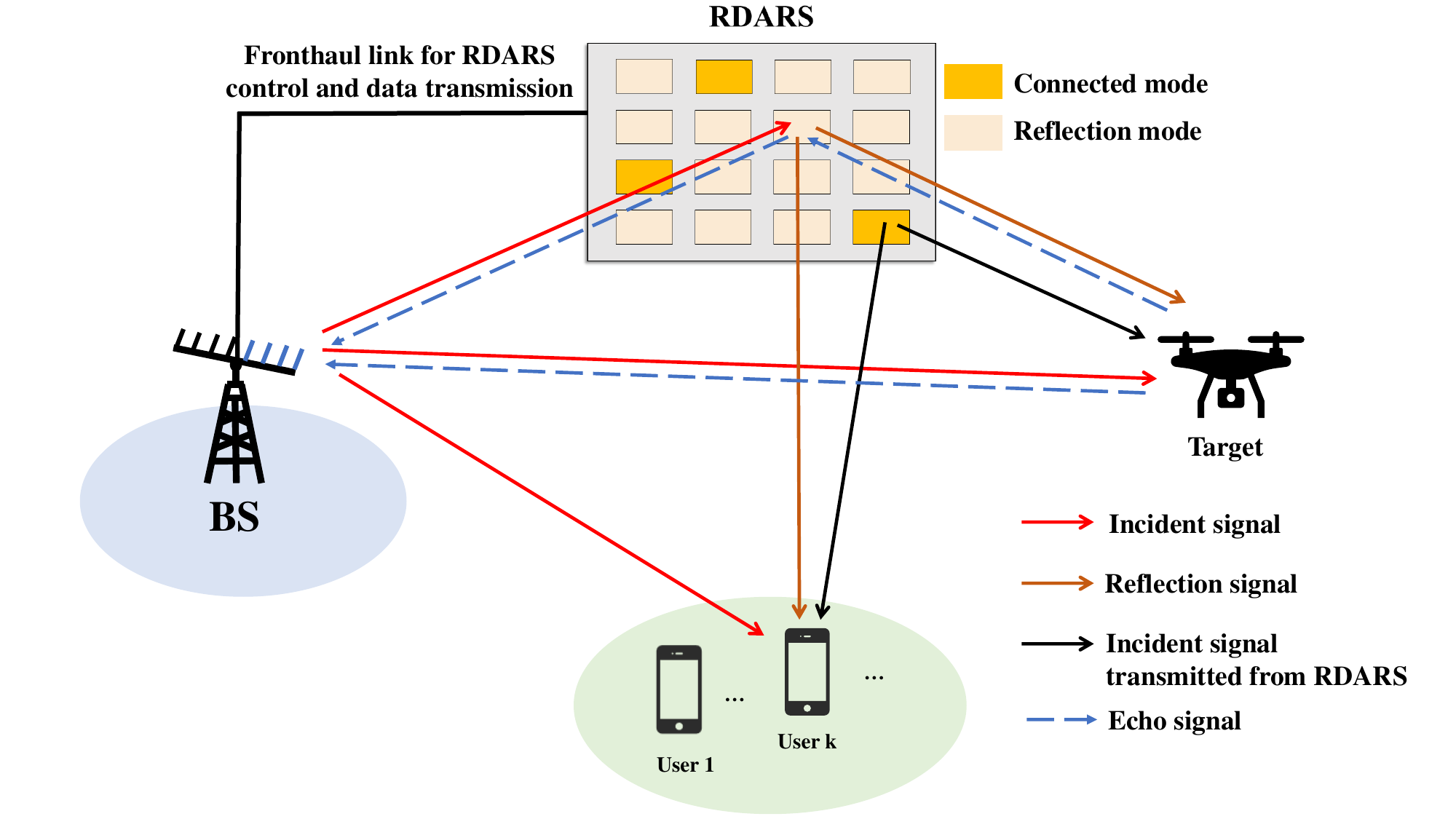}
    \vspace{-2mm}
    \caption{The system model of our RDARS-aided ISAC system.}
   \label{system_model}
    \end{figure} 
   
In the downlink communications, the BS and the connected elements of RDARS collaborate to serve users in a distributed antenna fashion.
Denote by $\mathbf{x}_{1}\in \mathbb{C}^{M} $ the signal transmitted from the BS and $\mathbf{x}_{2}\in \mathbb{C}^{a} $ the signal transmitted from the connected elements of RDARS. The combined transmitted signal can be written as
\begin{equation}\begin{aligned}
\mathbf{x}=\begin{bmatrix}\mathbf{x}_{1} \\\mathbf{x}_{2}\end{bmatrix}=\begin{bmatrix}\mathbf{F}_{1} \\\mathbf{F}_{2}\end{bmatrix}\mathbf{s}
\triangleq \mathbf{F}\mathbf{s},
\end{aligned}\end{equation}
where $\mathbf{s} \sim \mathcal {CN}(\mathbf{0}, {\mathbf{I}}_{K})$ is the transmit data symbol vector; $\mathbf{F}_{1}\in \mathbb{C}^{M\times K }$ and $\mathbf{F}_{2}\in \mathbb{C}^{a\times K }$ are the transmit beamforming matrices at the BS and the connected elements of RDARS, respectively; $\mathbf{F}$ is the compound transmit beamforming matrix.
\vspace{-5mm}    
\subsection{Communication Model} 
Denote by
$\mathbf {H}_{br}\in \mathbb {C}^{N\times M}$, 
$\mathbf {h}_{bu,k}\in \mathbb {C}^{M}$, and 
$\mathbf {h}_{ru,k}\in \mathbb {C}^{N}$ 
the channel between the BS and RDARS, 
the channel between the BS and user ${k}$, and 
the channel between the RDARS and user ${k}$, respectively.
We model the channel components in $\mathbf {H}_{br}$, $\mathbf {h}_{bu,k}$ and $\mathbf {h}_{ru,k}$ with a general Rician distribution. 
As an example, $\mathbf {H}_{br}$ can be written as
 \begin{equation}\begin{aligned}
 {\mathbf{H}_{br}} = \sqrt { \alpha_{\mathbf{H}_{br} }} (\sqrt {\kappa } \mathbf{H}_{br}^{{ LoS}} + \sqrt {1 - \kappa } \mathbf{H}_{br}^{{ NLoS}}),
 \label{H_br}\end{aligned}\end{equation} 
 where $\alpha_{\mathbf{H}_{br} }$ is the path loss coefficient and $\kappa$ is the Rician factor; $\mathbf{H}_{br}^{{ NLoS}}$ denotes the NLoS component, the element of which follows the Rayleigh distribution, i.e., $\mathbf{H}_{br}^{{ NLoS}}({i,j})\sim \mathcal {CN}({0},1)$; $\mathbf{H}_{br}^{{ LoS}}$ denotes the LoS component, which is given by
 \begin{equation}\begin{aligned}
&\mathbf{H}_{br}^{{ LoS}}=\mathbf{a}_{br2}(\theta_{br}^{A},\psi_{br}^{A})\mathbf{a}_{br1}^{T}(\theta_{br}^{D}),
 \end{aligned}\end{equation} 
 where $\mathbf{a}_{br1}(\theta_{br}^{D})$ is the transmit array response at the BS and $\mathbf{a}_{br2}(\theta_{br}^{A},\psi_{br}^{A})$ is the receive array response at the RDARS; $\theta_{br}^{D}$, $\theta_{br}^{A}$ and $\psi_{br}^{A}$ are the azimuth angle of departure (AoD) at the BS towards the RDARS, azimuth angle of arrival (AoA) at the RDARS, and elevation AoA at the RDARS, respectively.
The array response vectors for an $M$-elements ULA and a $N_{ 1}\times N_{2}$ UPA are respectively defined as
 \begin{align}
&\mathbf {a}_{\text{ULA}}(\theta)=\begin{bmatrix} 1, e^{-j\frac{2\pi}{\beta}d\sin(\theta)}, \cdots, e^{-j\frac{2\pi}{\beta}d(M-1)\sin(\theta)} 
\end{bmatrix}^{T}, 
\\ &\mathbf {a}_{\text{UPA}}(\theta,\psi)=[ 1, e^{-j\frac{2\pi}{\beta}d(\cos(\theta)\sin(\psi)+\sin(\psi))}, \cdots,\notag \\ & \quad \quad\quad\quad \quad\quad\quad e^{-j\frac{2\pi}{\beta}d((N_{1}-1)\cos(\theta)\sin(\psi)+(N_{2}-1)\sin(\psi))}
]^{T},
 \end{align}
 where $\theta$ and $\psi$ represent the azimuth AoD/AoA and elevation AoD/AoA, respectively; $\beta$ is the wavelength and $d$ denotes the spacing between two consecutive antenna elements.

 Based on the above settings, the received signal of the $k$th user can be written as
\begin{equation}\label{eq:yuk}
\begin{aligned}
{y}_{u,k}&=(\mathbf{h}_{bu,k}^{T}+\mathbf{h}_{ru,k}^{T}(\mathbf{I}-\mathbf{A})\boldsymbol{\Phi} \mathbf{H}_{br})\mathbf{x}_{1}
+\mathbf{h}_{ru,k}^{T}\mathbf{A}_{a}\mathbf{x}_{2}+{n}_{1,k}
\\ & ={\mathbf{h}}_{1,k}^{T}\mathbf{F}\mathbf{s}+{n}_{1,k},
\end{aligned}
\end{equation} 
where ${\mathbf{h}}_{1,k}=\begin{bmatrix}\mathbf{h}_{bu,k}^{T}+\mathbf{h}_{ru,k}^{T}(\mathbf{I}-\mathbf{A})\boldsymbol{\Phi} \mathbf{H}_{br}&\mathbf{h}_{ru,k}^{T}\mathbf{A}_{a}\end{bmatrix}^{T}$ is the composite channel from the BS to the $k$th user and  ${n}_{1,k} \sim \mathcal {CN}({0},\sigma _{1,k} ^{2} )$ is the additive white Gaussian noise (AWGN). 
It is worth noting that we introduce another selection matrix $\mathbf{A}_{a}\in \mathbb{R}^{N\times a }$ in \eqref{eq:yuk} to accommodate the different dimensions of $\mathbf{x}_{1}$ and $\mathbf{x}_{2}$.
Specifically, $\mathbf{A}_{a}\subseteq \mathbf{A}$ consists of all the columns of $\mathbf{A}$ that contain $1$ (i.e., the elements that operate in the connected mode).
Therefore, we have $\mathbf{A}=\mathbf{A}_{a}\mathbf{A}_{a}^{T}$. 
Notice that this communication model is general and can cover various scenarios by setting the channel parameters accordingly. For example, when the direct link between the BS and the user $k$ is blocked, the channel $\mathbf{h}_{bu,k}$ can be set to zero vector to model this blockage scenario, which is a widely discussed scenario for RIS-aided ISAC system \cite{10197455}.

The SINR of the $k$th user is given by
\begin{equation}\begin{aligned}
\gamma_{u,k}=\frac{|{\mathbf{h}}_{1,k}^{T}\mathbf{f}_{k}|^{2}}{\sum _{i \neq k}^{K}|{{\mathbf{h}}_{1,k}^{T}\mathbf{f}_{i}}|^{2}+{\sigma}_{1,k}^{2}} ,
\end{aligned}\end{equation} 
where $\mathbf{f}_{i}$ denotes the $i$th column of $\mathbf{F}$, $i=1,\cdots, K$.
\vspace{-3mm}
\subsection{Sensing Model}
In the RDARS-assisted ISAC system, communication signals reach the target in three different paths:
1) from the BS directly to the target,
2) from the BS to the RDARS reflection elements and then to the target,
and 3) from the RDARS connected elements to the target.
After being reflected by the target, these signals follow two distinct paths to the BS:
1) from the target directly to the BS,
2) from the target to the RDARS reflection elements and then to the BS.
Therefore, we can write the received echo signal at the BS as
\begin{equation}
\begin{aligned}
\mathbf{y}_{r}&\!=\!\alpha(\mathbf{h}_{bt}+\mathbf{H}_{br}^{{T}}(\mathbf{I}-\mathbf{A})\boldsymbol{\Phi} \mathbf{h}_{rt} )
[(\mathbf{h}_{bt}^{{T}}+\mathbf{h}_{rt}^{{T}}(\mathbf{I}-\mathbf{A})\boldsymbol{\Phi}\mathbf{H}_{br} )\mathbf{x}_{1}
\\&\quad+\mathbf{h}_{rt}^{{T}}\mathbf{A}_{a}\mathbf{x}_{2}]+\mathbf{n}_{2}
\\ &\!=\!\alpha{\mathbf{H}}_{2}\mathbf{F}\mathbf{s}+\mathbf{n}_{2},
\end{aligned}\end{equation} 
where $\alpha $ refers to the radar cross-section (RCS) with $\mathbb{E}\{|\alpha|^{2}\}=\sigma_{\alpha}^{2}$; $\mathbf{n}_{2} \sim \mathcal {CN}(\mathbf{0},\sigma _{2} ^{2} \mathbf{I}_{M})$ is the AWGN;
${\mathbf{H}}_{2}\triangleq(\mathbf{h}_{bt}+\mathbf{H}_{br}^{{T}}(\mathbf{I}-\mathbf{A})\boldsymbol{\Phi} \mathbf{h}_{rt} )
\allowbreak\begin{bmatrix}(\mathbf{h}_{bt}^{{T}}+\mathbf{h}_{rt}^{{T}}(\mathbf{I}-\mathbf{A})\boldsymbol{\Phi} \mathbf{H}_{br})&\mathbf{h}_{rt}^{{T}}\mathbf{A}_{a}\end{bmatrix}$ is the equivalent round-trip echo channel; $\mathbf{h}_{bt}\in \mathbb{C}^{M}$ and $\mathbf{h}_{rt}\in \mathbb{C}^{N }$ are the channels between the BS and target and that between the RDARS and target, respectively. $\mathbf{h}_{bt}$ and $\mathbf{h}_{rt}$ are respectively modeled as
\begin{equation}
 \mathbf {h}_{bt}=\sqrt { \alpha_{\mathbf{h}_{bt} }}\mathbf{a}_{bt}(\theta_{bt}^{D}),~~
\mathbf {h}_{rt}=\sqrt { \alpha_{\mathbf{h}_{rt} }}\mathbf{a}_{rt}(\theta_{rt}^{D},\psi_{rt}^{D}),
\end{equation}
where $\alpha_{\mathbf{h}_{bt}}$ and $\alpha_{\mathbf{h}_{rt} }$ are the path loss coefficients;
$\theta_{bt}^{D}$  is the azimuth AoD towards the target at the BS;
$\theta_{rt}^{D}$ and $\psi_{rt}^{D}$ are the azimuth AoD and elevation AoD towards the target at the RDARS, respectively.

Thanks to the availability of efficient channel estimation algorithms\cite{9854847,9366805,zhou2022channel}, we assume that full channel state information (CSI) is available at the BS. Upon receiving the echo signal, the BS applies the receive filter $\mathbf{w}\in \mathbb{C}^{M} $, yielding
\begin{equation}\begin{aligned}
  \mathbf{w}^{H}  { \mathbf{y}}_{r} =\alpha\mathbf{w}^{H}{\mathbf{H}}_{2}\mathbf{F}\mathbf{s}+\mathbf{w}^{H}{\mathbf{n}}_{2}.
\label{receive_beamformer}\end{aligned}\end{equation}
Therefore, the radar output SNR can be written as
\begin{equation}\begin{aligned}
\gamma_{t}=\frac{\sigma_{\alpha}^{2}\mathbf{w}^{H}{\mathbf{H}}_{2}\mathbf{F}\mathbf{F}^{H}{\mathbf{H}}_{2}^{H}\mathbf{w}}{{\sigma}_{2}^{2}\mathbf{w}^{H}\mathbf{w} }.
\end{aligned}\end{equation}

\vspace{-4mm} 
\subsection{Problem Formulation}
In this paper, our primary focus centers on the enhancement of the target detection probability while maintaining the multi-user communication performance. 
As the performance of target detection is closely related to the radar output SNR, we tackle this challenge by formulating a radar output SNR maximization problem under stringent communication requirements. This problem involves a joint optimization process that encompasses several key components, including the receive filter $\mathbf{w}$ at the BS, the compound transmit beamforming matrix $\mathbf{F}$, the reflection matrix $\boldsymbol{\Phi}$ of the RDARS, and the two related selection matrices $\mathbf{A}$ and $\mathbf{A}_{a}$. Specifically, the optimization problem is formulated as
\begin{subequations}\begin{align}\text {(P1)} \quad& \max _{ \mathbf{F}, \boldsymbol{\Phi},\mathbf{w},\mathbf{A}, \mathbf{A}_{a}}~\frac{\sigma_{\alpha}^{2}\mathbf{w}^{H}{\mathbf{H}}_{2}\mathbf{F}\mathbf{F}^{H}{\mathbf{H}}_{2}^{H}\mathbf{w}}{{\sigma}_{2}^{2}\mathbf{w}^{H}\mathbf{w} }
\\& {\qquad  ~\text {s.t.}~} 
     \gamma_{u,k}\geq \overline{\gamma}_{k}, \forall {k} ,\label{com_constraint}
      \\&\hphantom {\qquad  ~\text {s.t.}~} {\Vert\mathbf{F}\Vert}_{F}^{2} \leq P, \label{power_constraint}
     \\&\hphantom {\qquad  ~\text {s.t.}~} 
       \boldsymbol {\Phi }= \mathop {\mathrm {diag}}\nolimits \{\phi  _{1},\cdots,\phi _{N}\},\label{RIS_constraint1}
       \\&\hphantom {\qquad  ~\text {s.t.}~} |\phi  _{{i}}|=1, \quad i\in \{1,\cdots,N\},\label{RIS_constraint2}
       \\&\hphantom {\qquad  ~\text {s.t.}~} 
       \mathbf{A}=\mathbf{A}_{a}\mathbf{A}_{a}^{T},\label{A_constraint1}
       \\&\hphantom {\qquad  ~\text {s.t.}~} {\mathbf{A}}(i,i)\in\{0,1\}, \quad i\in \{1,\cdots,N\},\label{A_constraint2}
     \end{align}  \end{subequations} 
where \eqref{com_constraint} is the communication SINR constraint for each user and $\overline{\gamma}_{k}$ is the minimum SINR threshold for the $k$th user; 
\eqref{power_constraint} is the total power constraint; 
\eqref{RIS_constraint2} is the unit-modulus constraint for each RDARS element; 
\eqref{A_constraint1} and \eqref{A_constraint2} are the constraints for the two selection matrices. 
It is worth noting that problem (P1) is a non-convex problem, which is challenging to solve due to the following reasons.
First, the optimization variables are tightly coupled in the objective function and the objective function does not have an explicit form with respect to (w.r.t.) the variables $\boldsymbol{\Phi}$, $\mathbf{A}$, and $\mathbf{A}_{a}$, which further increases the solving difficulty.
Second, the unit-modulus constraint in \eqref{RIS_constraint2} is non-convex.
Third, \eqref{A_constraint2} is a non-convex binary integer constraint, which is NP-hard. Roughly speaking, the optimizations of the beamforming matrix F, the reflection matrix $\boldsymbol{\Phi}$ and the selection matrices would bring distribution gain, reflection gain and selection gain, respectively. This will be clearly demonstrated later in the simulation results.

\vspace{-2mm}
\section{Joint Beamforming And Mode Selection For RDARS-Aided ISAC}
This section presents our approach for designing joint beamforming and mode selection in the context of RDARS-aided ISAC.
To address the intricacies of the non-convex problem (P1), we introduce a streamlined algorithm that seamlessly integrates elements from the AO, penalty, and MM methods.
Our approach begins with AO, decoupling the complex original problem into a set of sub-problems that are amenable to analysis and optimization. Subsequently, we utilize both MM and penalty methods in an iterative fashion to tackle each of these sub-problems.

\vspace{-2mm} 
\subsection{Receive Filter Optimization}
We first derive the optimal receive filter $\mathbf{w}^{\star}$. For fixed $\mathbf{F}$, $\boldsymbol{\Phi}$, $\mathbf{A}$ and $\mathbf{A}_{a}$, the optimization problem w.r.t. the receive filter $\mathbf{w}$ can be written as
\begin{equation}\begin{aligned}\text {(P2)} &\quad \max _{ \mathbf{w} }~ 
\frac{\sigma_{\alpha}^{2}\mathbf{w}^{H}{\mathbf{H}}_{2}\mathbf{F}\mathbf{F}^{H}{\mathbf{H}}_{2}^{H}\mathbf{w}}{{\sigma}_{2}^{2}\mathbf{w}^{H}\mathbf{w} }.
 \label{optimize_w}  \end{aligned} \end{equation}
(P2) is a typical Rayleigh quotient problem, the optimal solution of which is the eigenvector corresponding to the largest eigenvalue of ${\mathbf{H}}_{2}\mathbf{F}\mathbf{F}^{H}{\mathbf{H}}_{2}^{H}$.

\vspace{-2mm} 
\subsection{Transmit Beamforming Matrix Optimization}
With given $\mathbf{w}$, $\boldsymbol{\Phi}$, $\mathbf{A}$ and $\mathbf{A}_{a}$, the optimization problem w.r.t the transmit beamforming matrix $\mathbf{F}$ is given by 
\begin{subequations}\begin{align}\text {(P3)} &\quad \max _{ \mathbf{F} }~ \frac{\sigma_{\alpha}^{2}\mathbf{w}^{H}{\mathbf{H}}_{2}\mathbf{F}\mathbf{F}^{H}{\mathbf{H}}_{2}^{H}\mathbf{w}}{{\sigma}_{2}^{2}\mathbf{w}^{H}\mathbf{w} }
\\& {\qquad  ~\text {s.t.}~} 
     \frac{|{\mathbf{h}}_{1,k}^{T}\mathbf{f}_{k}|^{2}}{\sum _{i \neq k}^{K}|{{\mathbf{h}}_{1,k}^{T}\mathbf{f}_{i}}|^{2}+{\sigma}_{1}^{2}}\geq \overline{\gamma}_{k}, \forall {k} ,\label{sinr}
      \\&\hphantom {\qquad  ~\text {s.t.}~} {\Vert\mathbf{F}\Vert}_{F}^{2} \leq P.
     \end{align}  \end{subequations} 
 
It can be found that the transmit beamforming matrix $\mathbf{F}$ has different forms in the objective function and constraint \eqref{sinr}, significantly complicating the problem-solving process.
To render the problem more manageable, we first undertake some transformations on both the objective function and the constraint \eqref{sinr}. Specifically, by using the properties $\mathrm{tr}(\mathbf{A}\mathbf{B}\mathbf{C}\mathbf{D})=\mathrm{tr}(\mathbf{C}\mathbf{D}\mathbf{A}\mathbf{B})$ and  $\mathrm{tr}(\mathbf{C}\mathbf{D}\mathbf{A}\mathbf{B})=\mathrm{vec}(\mathbf{C}^{H})^{H}(\mathbf{B}^{T}\otimes\mathbf{D})\mathrm{vec}(\mathbf{A})$, the objective function can be transformed to
\begin{equation}\begin{aligned}
\frac{\sigma_{\alpha}^{2}\mathbf{w}^{H}{\mathbf{H}}_{2}\mathbf{F}\mathbf{F}^{H}{\mathbf{H}}_{2}^{H}\mathbf{w}}{{\sigma}_{2}^{2}\mathbf{w}^{H}\mathbf{w} }=\mathbf{f}^{H}\mathbf{C}\mathbf{f},
\end{aligned}\end{equation}
where $\mathbf{f}\triangleq\mathrm{vec}(\mathbf{F})$ and $\mathbf{C}=\frac{{\sigma}_{2}^{2}\mathbf{I}_{K}\otimes({\mathbf{H}}_{2}^{H}\mathbf{w}\mathbf{w}^{H}{\mathbf{H}}_{2})}{{\sigma}_{2}^{2}\mathbf{w}^{H}\mathbf{w}}$.
Note that although $\mathbf{f}^{H}\mathbf{C}\mathbf{f}$ is a convex function w.r.t $\mathbf{f}$, maximizing a convex function itself is a non-convex problem. 
To navigate the complexity and arrive at a more manageable solution, we employ the MM technique to find a more tractable surrogate function. Specifically, based on the first-order Taylor expansion, we have \cite{sun2016majorization}
\begin{equation}\begin{aligned}
\mathbf{f}^{H}\mathbf{C}\mathbf{f}\geq \mathbf{f}_{t}^{H}\mathbf{C}\mathbf{f}_{t}+2\Re \{\mathbf{f}_{t}^{H}\mathbf{C}(\mathbf{f}-\mathbf{f}_{t})\},
\end{aligned}\end{equation}
where $\mathbf{f}_{t}$ is the optimal solution of $\mathbf{f}$ at the $t$th iteration. Since $\mathbf{f}_{t}^{H}\mathbf{C}\mathbf{f}_{t}$ is the known term, we just need to optimize $\Re \{\mathbf{f}_{t}^{H}\mathbf{C}\mathbf{f}\}$, which is a linear function.
For constraint \eqref{sinr}, we first express $\mathbf{f}_{i}$ as $\mathbf{f}_{i}=\Pi_{j}\mathbf{f}$, where $\Pi_{j} \in\mathbb{C}^{(M+a)\times K(M+a)}$ is the permutation matrix. Then, constraint \eqref{sinr} can be rewritten as 
\begin{equation}\begin{aligned}
\sqrt{1+\overline{\gamma}_{k}}{\mathbf{h}}_{1,k}^{T}\Pi_{j}\mathbf{f}\geq \sqrt{\overline{\gamma}_{k}}\lVert[\mathbf{B}_{k}\mathbf{f}, \sigma_{1}]\rVert,
\label{SOC}\end{aligned}\end{equation}
where $\mathbf{B}_{k}=\mathbf{I}_{K}\otimes \mathbf{h}_{1,k}^{T}$. The constraint \eqref{SOC} is a convex second-order cone (SOC) constraint. Based on the above transformations, problem (P3) can be reformulated as 
\begin{subequations}\begin{align}\label{optimize_f}\text {(P4)} &\quad \max _{ \mathbf{f} }~ \Re \{\mathbf{f}_{t}^{H}\mathbf{C}\mathbf{f}\}
\\& {\qquad  ~\text {s.t.}~} 
    \sqrt{1+\overline{\gamma}_{k}}|{\mathbf{h}}_{1,k}^{T}\Pi_{k}\mathbf{f}|\geq \sqrt{\overline{\gamma}_{k}}\lVert[\mathbf{B}_{k}\mathbf{f}, \sigma_{1}]\rVert, \forall {k} ,
      \\&\hphantom {\qquad  ~\text {s.t.}~} {\Vert\mathbf{f}\Vert}^{2} \leq P,
    \end{align}  \end{subequations} 
which is a second-order cone programming (SOCP) problem and can be solved by the existing efficient convex toolboxes, such as CVX\cite{CVX}.

\vspace{-2mm} 
\subsection{Reflection Matrix and Selection Matrices Optimization}
The optimization variables $\boldsymbol{\Phi}$ and $\mathbf{A}$ share many similarities. Specifically, $\boldsymbol{\Phi}$ and $\mathbf{A}$ are both diagonal matrices and for each element on the diagonal, the modulus is always one. In addition, for  $\boldsymbol{\Phi}$ and $\mathbf{A}$, the equations $(\mathrm{diag}(\boldsymbol{\Phi}))^{H}\mathrm{diag}(\boldsymbol{\Phi})=N$ and $(\mathrm{diag}(\mathbf{A}))^{T}\mathrm{diag}(\mathbf{A})=a$ hold due to the unit-modulus property. Though $\mathbf{A}_{a}$ is not a square matrix, it still has the similar properties
as $\boldsymbol{\Phi}$ and $\mathbf{A}$, such as $(\mathrm{vec}(\mathbf{A}_{a}))^{T}\mathrm{vec}(\mathbf{A}_{a})=a$. Therefore, we shall optimize the three optimization variables together.

After obtaining $\mathbf{w}$ and $\mathbf{f}$, the optimization problem w.r.t. the reflection matrix $\boldsymbol{\Phi}$ and the two selection matrices $\mathbf{A}$ and $\mathbf{A}_{a}$ is formulated as
\begin{subequations}\begin{align}\text {(P5)} &\quad \max _{   \boldsymbol{\Phi}, \mathbf{A} ,\mathbf{A}_{a} }~ \frac{\sigma_{\alpha}^{2}\mathbf{w}^{H}{\mathbf{H}}_{2}\mathbf{F}\mathbf{F}^{H}{\mathbf{H}}_{2}^{H}\mathbf{w}}{{\sigma}_{2}^{2}\mathbf{w}^{H}\mathbf{w} }
\\& {\qquad  ~\text {s.t.}~} 
     \frac{|{\mathbf{h}}_{1,k}^{T}\mathbf{f}_{k}|^{2}}{\sum _{i \neq k}^{K}|{{\mathbf{h}}_{1,k}^{T}\mathbf{f}_{i}}|^{2}+{\sigma}_{1}^{2}}\geq \overline{\gamma}_{k}, \forall {k}, \label{14b}
     \\&\hphantom {\qquad  ~\text {s.t.}~} 
       \boldsymbol {\Phi }= \mathop {\mathrm {diag}}\nolimits \{\phi  _{1},\cdots,\phi _{N}\},\label{16c}
       \\&\hphantom {\qquad  ~\text {s.t.}~} |\phi  _{{i}}|=1, \quad i\in \{1,\cdots,N\},\label{16d}
       \\&\hphantom {\qquad  ~\text {s.t.}~} \mathbf{A}=\mathbf{A}_{a}\mathbf{A}_{a}^{T},\label{16e}
       \\&\hphantom {\qquad  ~\text {s.t.}~} {\mathbf{A}}(i,i)\in\{0,1\}, \quad i\in \{1,\cdots,N\}.\label{16f}
     \end{align}  \end{subequations} 
Dealing directly with this sub-problem remains a challenging task. The optimization variables in the objective function, as well as constraints \eqref{14b} and \eqref{16e}, are intricately intertwined. On the other hand, constraints \eqref{16c}, \eqref{16d}, and \eqref{16f} operate independently.
To create a unified framework, we employ the penalty technique to consolidate constraints \eqref{14b} and \eqref{16e} into penalty terms. To achieve this, we introduce auxiliary variables $s_{k,i},\forall{k}, {i}$ to reconfigure constraint \eqref{14b} as 
\begin{subequations}\begin{align}
&\frac{|s_{k,k}|^{2}}{\sum _{i \neq k}^{K}|s_{k,i}|^{2}+{\sigma}_{1}^{2}}\geq \overline{\gamma}_{k}, \forall {k}, 
\\&{\mathbf{h}}_{1,k}^{T}\mathbf{f}_{i}=s_{k,i},\forall{k},\forall{i}.\label{17b}
\end{align}  \end{subequations} 
By adding the equality constraints \eqref{16e} and \eqref{17b} to the objective function, the corresponding optimization problem can be rewritten as
\begin{subequations}\begin{align}\text {(P6)} &\quad \min _{  \boldsymbol{\Phi}, \mathbf{A} ,\mathbf{A}_{a},s_{k,i}}~ -\frac{\sigma_{\alpha}^{2}\mathbf{w}^{H}{\mathbf{H}}_{2}\mathbf{F}\mathbf{F}^{H}{\mathbf{H}}_{2}^{H}\mathbf{w}}{{\sigma}_{2}^{2}\mathbf{w}^{H}\mathbf{w} }+
\notag\\&\qquad\qquad\qquad \frac{1}{2\rho_{1}}\sum _{k =1}^{K}\sum _{i=1}^{K}|{\mathbf{h}}_{1,k}^{T}\mathbf{f}_{i}-s_{k,i}|^{2}\notag\\&\qquad \qquad\qquad+\frac{1}{2\rho_{2}}\lVert  \mathbf{A}-\mathbf{A}_{a}\mathbf{A}_{a}^{T}\rVert_{F}^{2}  \label{obj_lag}
 \\& {\qquad  ~\text {s.t.}~} 
    \frac{|s_{k,k}|^{2}}{\sum _{i \neq k}^{K}|s_{k,i}|^{2}+{\sigma}_{1}^{2}}\geq \overline{\gamma}_{k}, \forall {k}, 
     \\&\hphantom {\qquad  ~\text {s.t.}~} 
       \boldsymbol {\Phi }= \mathop {\mathrm {diag}}\nolimits \{\phi  _{1},\cdots,\phi _{N}\},
       \\&\hphantom {\qquad  ~\text {s.t.}~} |\phi  _{{i}}|=1, \quad i\in \{1,\cdots,N\},
       \\&\hphantom {\qquad  ~\text {s.t.}~} {\mathbf{A}}(i,i)\in\{0,1\}, \quad i\in \{1,\cdots,N\},
     \end{align}  \end{subequations} 
where $\rho_{1}>0 $ and $\rho_{2}>0$ are the penalty coefficients, which are used to penalize the violation of the equality constraints \eqref{16e} and \eqref{17b}. By gradually decreasing the value of $\rho_{1}$ and $\rho_{2}$ ( $\frac{1}{\rho_{1}}\rightarrow \infty$ and $\frac{1}{\rho_{2}}\rightarrow \infty$), the obtained optimal solution by solving problem (P6) is guaranteed to satisfy the equality constraints \eqref{16e} and \eqref{17b}. In general, it is desirable to initialize $\rho_{1}$ and $\rho_{2}$ to be a slightly larger number, such that the penalized objective function is dominated by the original objective function rather than the penalty terms\cite{9133435}. 

With fixed $\rho_{1}$ and $\rho_{2}$, problem (P6) is still non-convex. However, based on the above transformations, for each optimization variable, though it is coupled with other optimization variables in the objective function, the involved constraint is independent. With this observation, we propose to employ the block coordinate descent (BCD) method to problem (P6) by iteratively solving each block while fixing other blocks. To make the objective function more explicit, we first prove the following lemma and use it to express the objective function w.r.t. $\boldsymbol{\Phi}$, $\mathbf{A}$ and $\mathbf{A}_{a}$, respectively.

\begin{lemma} 
 By defining $\boldsymbol{\phi}\!=\![\phi  _{1},\cdots,\phi _{N}]^{T}$, $\mathbf{a}\!=\![a_{1}, \cdots,a_{N}]^{T}$ and $ \mathbf{a}_{a}=[\mathbf{a}_{a1}^{T},\cdots,\mathbf{a}_{aa}^{T}]^{T}$, where $\mathbf{a}_{ai}$ is the $i$th column of $\mathbf{A}_{a}$, the objective function can be explicitly expressed w.r.t. $\boldsymbol{\phi}$, $\mathbf{a}$ and $\mathbf{a}_{a}$ as
\begin{subequations}\begin{align}
f_{\mathrm{obj},\boldsymbol{\phi}}=&-({r}_{1}+2\Re\{\mathbf{r}_{2}^{H}\boldsymbol{\phi}\}+2\Re\{\mathbf{r}_{3}^{H}(\boldsymbol{\phi}\otimes\boldsymbol{\phi})\}
+\boldsymbol{\phi}^{H}\mathbf{R}_{4}\boldsymbol{\phi}\notag
\\& +2\Re\{(\boldsymbol{\phi}\otimes\boldsymbol{\phi})^{H}\mathbf{R}_{5}\boldsymbol{\phi}\}
+(\boldsymbol{\phi}\otimes\boldsymbol{\phi})^{H}\mathbf{R}_{6}(\boldsymbol{\phi}\otimes\boldsymbol{\phi}))
\notag\\& +\frac{1}{2\rho_{1}}(r_{7}+2\Re\{\mathbf{r}_{8}^{H}\boldsymbol{\phi}\}+\boldsymbol{\phi}^{T}\mathbf{R}_{9}\boldsymbol{\phi}^{\ast}),
\label{objective_function1}
\\
f_{\mathrm{obj},\mathbf{a}}=&-(r_{10}+2\Re\{\mathbf{r}_{11}^{H}\mathbf{a}\}
+2\Re\{\mathbf{r}_{12}^{H}(\mathbf{a}\!\otimes\!\mathbf{a})\}
+\mathbf{a}^{T}\mathbf{R}_{13}\mathbf{a}\notag
\\& 
+2\Re\{(\mathbf{a}\!\otimes\!\mathbf{a})^{T}\mathbf{R}_{14}\mathbf{a}\}
+(\mathbf{a}\!\otimes\!\mathbf{a})^{T}\mathbf{R}_{15}(\mathbf{a}\!\otimes\!\mathbf{a}))\notag
\!+\!\frac{1}{2\rho_{1}}\\&(r_{16}+2\Re\{\mathbf{r}_{17}^{H}\mathbf{a}\}+\mathbf{a}^{T}\mathbf{R}_{18}\mathbf{a})+\frac{1}{2\rho_{2}}(\mathbf{r}_{19}^{T}\mathbf{a}+r_{20}),
\label{objective_function2}
\\
f_{\mathrm{obj},\mathbf{a}_{a}}=&-(r_{21}+2\Re\{{\mathbf{r}}_{22}^{H}\mathbf{a}_{a}+\mathbf{a}_{a}^{T}{\mathbf{R}}_{23}\mathbf{a}_{a}\})+\frac{1}{2\rho_{1}}(r_{24}+\notag
\\&2\Re\{\mathbf{r}_{25}^{H}\mathbf{a}_{a}\}+\mathbf{a}_{a}^{T}\mathbf{R}_{26}\mathbf{a}_{a})\!+\!\frac{1}{\rho_{2}}(a-\mathbf{a}_{a}^{T}(\mathbf{I}_{a}\otimes\mathbf{A})\mathbf{a}_{a}),\label{objective_function3}
\end{align}\end{subequations}
 respectively, where ${r}_{1}$, $\mathbf{r}_{2}$, $\mathbf{r}_{3}$, $\mathbf{R}_{4},\cdots, {r}_{24}$, $\mathbf{r}_{25}$ and $\mathbf{R}_{26}$  are constants. In particular, ${r}_{1}$, $\mathbf{r}_{2}$, $\mathbf{r}_{3}$, $\mathbf{R}_{4},\cdots, \mathbf{R}_{6}$, ${r}_{7}$, $\mathbf{r}_{8}$, $\mathbf{R}_{9}$, $\mathbf{r}_{19}$ and ${r}_{20}$, are defined in \eqref{r1_r2}, \eqref{r3_r4}, \eqref{r5_r6}, \eqref{r7}, \eqref{r8}, \eqref{r9} and \eqref{r19_r20}, respectively, and the other constants share the similar form with ${r}_{1}$, $\mathbf{r}_{2}$, $\mathbf{r}_{3}$, $\mathbf{R}_{4},\cdots, \mathbf{R}_{6}$, ${r}_{7}$, $\mathbf{r}_{8}$ and $\mathbf{R}_{9}$, thus the explicit forms of which are omitted.

\end{lemma}
\begin{proof} 
See Appendix \ref{sec:AppA}.
\end{proof} 

To solve problem (P6), we can iteratively update the following four blocks.

\subsubsection{Update \texorpdfstring{$s_{k,i}$}{\texttwoinferior}} With fixed $\boldsymbol{\phi}$, $\mathbf{a}$ and $\mathbf{a}_{a}$, the sub-problem w.r.t. $s_{k,i}$ is given by
\begin{subequations}\begin{align}\text {(P7)} &\quad \min _{  s_{k,i}}~ \sum _{k =1}^{K}\sum _{i=1}^{K}|{\mathbf{h}}_{1,k}^{T}\mathbf{f}_{i}-s_{k,i}|^{2}
\\& {\quad  ~\text {s.t.}~} 
    \frac{|s_{k,k}|^{2}}{\sum _{i \neq k}^{K}|s_{k,i}|^{2}+{\sigma}_{1,k}^{2}}\geq \overline{\gamma}_{k}, \forall {k}. 
     \end{align}
\end{subequations} 

The objective function comprises a total of $K$ terms, with each of these $K$ terms being independent from one another.
Furthermore, for each $k$, the variables $\{ s_{k,i},  \forall {i}\}$ are separable within the constraints.
Hence, we have the flexibility to tackle $K$ sub-problems independently and in parallel. Specifically, the sub-problem corresponding to the $k$th block $\{ s_{k,i},  \forall {i}\}$ is given by
\vspace{-2mm}
\begin{subequations}\begin{align}\text {(P8)} &\quad \min _{ s_{k,i}}~ \sum _{i=1}^{K}|{\mathbf{h}}_{1,k}^{T}\mathbf{f}_{i}-s_{k,i}|^{2}
\\ & \quad{ ~\text {s.t.}~}\frac{|s_{k,k}|^{2}}{\sum _{i \neq k}^{K}|s_{k,i}|^{2}+{\sigma}_{1,k}^{2}}\geq \overline{\gamma}_{k}.\label{22b}\end{align}  \end{subequations} 
Problem (P8) is a non-convex  quadratically constrained quadratic
programming (QCQP) problem, which is difficult to tackle. However, \cite{boyd_vandenberghe_2004} shows that  the strong duality will hold for problem (P8) when the Slater's condition requirement is met. If so, the duality gap between problem (P8) and its dual problem becomes zero. Therefore, the Lagrangian dual method can be utilized to get the solution.
By introducing the dual variable  $\mu_{k}\geq0$, the Lagrange function of problem (P8) is given by
\begin{equation}\begin{aligned}
&\mathcal{L}(s_{k,i},\mu_{k})\\&=\sum _{i=1}^{K}\!|{\mathbf{h}}_{1,k}^{T}\mathbf{f}_{i}-s_{k,i}|^{2}\!+\!\mu_{k}(\overline{\gamma}_{k}(\sum _{i \neq k}^{K}\!|s_{k,i}|^{2}\!+\!{\sigma}_{1,k}^{2})\!-\!|s_{k,k}|^{2}).
\end{aligned}\end{equation}
Then the  dual function of problem (P8) can  be obtained as $f_{dual}(\mu_{k})=\min\limits_{ s_{k,i}}~ \mathcal{L}(s_{k,i},\mu_{k})$. It should be noted that to make $f_{dual}(\mu_{k})$ bounded, $\mu_{k}$ should satisfy $0\leq\mu_{k}<1$. The proof is omitted here and the interested readers can refer to \cite{9913311} for more details.
Taking the first-order derivative of $f_{dual}(\mu_{k})$ to zero, we can obtain the optimal $s_{k,i}^{\star}$ as
\begin{equation} 
\left\{  
             \begin{aligned}  
            & s_{k,i}^{\star}=\frac{{\mathbf{h}}_{1,k}^{T}\mathbf{f}_{i}}{1+\mu_{k}\overline{\gamma}_{k}}, i\neq k, \forall {i}, \\  
            & s_{k,i}^{\star}=\frac{{\mathbf{h}}_{1,k}^{T}\mathbf{f}_{k}}{1-\mu_{k}}, i=k.
             \end{aligned}  \right. 
\label{opt_ski}\end{equation} 
Next, we aim to find the optimal dual variable  $\mu_{k}^{\star}$. For optimal $s_{k,i}^{\star}$  and $\mu_{k}^{\star}$, the following Slater's condition should be satisfied.
\begin{equation}\begin{aligned}
\mu_{k}^{\star}(\overline{\gamma}_{k}(\sum _{i \neq k}^{K}|s_{k,i}^{\star}(\mu_{k}^{\star})|^{2}+{\sigma}_{1,k}^{2})-|s_{k,k}^{\star}(\mu_{k}^{\star})|^{2})=0.
\end{aligned}\end{equation}
If constraint \eqref{22b} strictly holds, i.e., 
\begin{equation}\begin{aligned}
\overline{\gamma}_{k}(\sum _{i \neq k}^{K}|s_{k,i}^{\star}(\mu_{k}^{\star})|^{2}+{\sigma}_{1,k}^{2})-|s_{k,k}^{\star}(\mu_{k}^{\star})|^{2})<0,
\end{aligned}\end{equation}
we have $\mu_{k}^{\star}=0$ and the optimal $s_{k,i}^{\star}={\mathbf{h}}_{1,k}^{T}\mathbf{f}_{i}, \forall {i}, \forall {k}$, otherwise, we have $0<\mu_{k}^{\star}<1$. Substituting \eqref{opt_ski} into \eqref{22b}, the equality constraint is given by
\begin{equation}\begin{aligned}
\overline{\gamma}_{k}(\sum _{i \neq k}^{K}|\frac{{\mathbf{h}}_{1,k}^{T}\mathbf{f}_{i}}{1+\mu_{k}^{\star}\overline{\gamma}_{k}}|^{2}+{\sigma}_{1,k}^{2})-|\frac{{\mathbf{h}}_{1,k}^{T}\mathbf{f}_{k}}{1-\mu_{k}^{\star}}|^{2}=0.
\end{aligned}\end{equation}
By defining $f_{\mu_{k}}(\mu_{k})\triangleq\overline{\gamma}_{k}(\sum _{i \neq k}^{K}|\frac{{\mathbf{h}}_{1,k}^{T}\mathbf{f}_{i}}{1+\mu_{k}\overline{\gamma}_{k}}|^{2}+{\sigma}_{1,k}^{2})-|\frac{{\mathbf{h}}_{1,k}^{T}\mathbf{f}_{k}}{1-\mu_{k}}|^{2}$, we can easily see that $f_{\mu_{k}}(\mu_{k})$ is a monotonically
decreasing function w.r.t $\mu_{k}$ for $0<\mu_{k}<1$. Hence, the optimal $\mu_{k}^{\star}$ can be efficiently obtained by the one-dimension search technique, such as the bisection search \cite{shi2011iteratively}.

\subsubsection{Update \texorpdfstring{$\boldsymbol{\phi}$}{\texttwoinferior}}  With fixed $s_{k,i}$, $\mathbf{a}$ and $\mathbf{a}_{a}$, the sub-problem w.r.t. $\boldsymbol{\phi}$ is formulated as
\begin{subequations}\begin{align}\text {(P9)} &\quad \min _{  \boldsymbol{\phi} }~ f_{\mathrm{obj},\boldsymbol{\phi}}\label{obj_phi}
\\& {\qquad  ~\text {s.t.}~} 
    |\phi  _{{i}}|=1, \quad i\in \{1,\cdots,N\}.
     \end{align}  \end{subequations} 
It can be readily found that the highest order of the optimization variable $\boldsymbol{\phi}$ is up to four and the objective function \eqref{obj_phi}  involves the  Kronecker product operations, which makes the objective function \eqref{obj_phi} so complicated that we can't directly solve problem (P9). In the following, we first employ the MM technique to find a more tractable surrogate function.
Define $\tilde{\boldsymbol{\phi}}\triangleq[\boldsymbol{\phi}^{T},\boldsymbol{\phi}^{T}\otimes\boldsymbol{\phi}^{T}]^{T}$ and $\tilde{\mathbf{R}}\triangleq\begin{bmatrix}\mathbf{R}_{4}&\mathbf{R}_{5}^{H}\\\mathbf{R}_{5}&\mathbf{R}_{6}\end{bmatrix}$. According to the definitions of $\mathbf{R}_{4}$, $\mathbf{R}_{5}$ and $\mathbf{R}_{6}$, we can easily observe that $\tilde{\mathbf{R}}$ is a positive semidefinite matrix. Assuming that the optimal solution at the $t$th iteration is $\tilde{\boldsymbol{\phi}}_{t}\triangleq[\boldsymbol{\phi}_{t}^{T},\boldsymbol{\phi}_{t}^{T}\otimes\boldsymbol{\phi}_{t}^{T}]^{T}$, then based on the first-order Taylor expansion, an upper bound of the term $-\tilde{ \boldsymbol {\phi}}^{\mathrm{H}} \tilde{ \mathbf {R}}\tilde { \boldsymbol {\phi}}$ can be derived as 
\begin{equation} -\tilde{ \boldsymbol {\phi}}^{\mathrm{H}} \tilde{ \mathbf {R}}\tilde { \boldsymbol {\phi}}\leq -2{\Re}\{\tilde{\boldsymbol {\phi}}_{t}^{\mathrm{H}}\tilde{ \mathbf {R}}\tilde{\boldsymbol {\phi}}\}+\tilde{\boldsymbol {\phi}}_{t}^{\mathrm{H}}\tilde{ \mathbf {R}}{\tilde{\boldsymbol {\phi}}}_{t}.\end{equation}
 Since the term $\tilde{\boldsymbol {\phi}}_{t}^{\mathrm{H}}\tilde{ \mathbf {R}}{\tilde{\boldsymbol {\phi}}}_{t}$ is irrelevant to $\boldsymbol {\phi}$, we only need to process the term $-2{\Re}\{\tilde{\boldsymbol {\phi}}_{t}^{\mathrm{H}}\tilde{ \mathbf {R}}\tilde{\boldsymbol {\phi}}\}$. Combining the term $-2{\Re}\{\tilde{\boldsymbol {\phi}}_{t}^{\mathrm{H}}\tilde{ \mathbf {R}}\tilde{\boldsymbol {\phi}}\}$ with the second and the third term of the objective function \eqref{obj_phi}, we have 
\begin{equation}\begin{aligned}
&-\left(2\Re\{\mathbf{r}_{2}^{H}\boldsymbol{\phi}\}+2\Re\{\mathbf{r}_{3}^{H}(\boldsymbol{\phi}\otimes\boldsymbol{\phi})\}+2{\Re}\{\tilde{\boldsymbol {\phi}}_{t}^{\mathrm{H}}\tilde{ \mathbf {R}}\tilde{\boldsymbol {\phi}}\}\right)
\\&=-2\Re\{\mathbf{q}_{1}^{H}\boldsymbol{\phi}+\mathbf{q}_{2}^{H}(\boldsymbol{\phi}\otimes\boldsymbol{\phi})\}
\\&=-2\Re\{\mathbf{q}_{1}^{H}\boldsymbol{\phi}+\boldsymbol{\phi}^{H}\mathbf{Q}_{2}\boldsymbol{\phi}^{\ast}\}
,\end{aligned}\end{equation}
where $\mathbf{q}_{1}=\mathbf{r}_{2}+\mathbf{R}_{4}^{H}\boldsymbol {\phi}_{t}+\mathbf{R}_{5}^{H}(\boldsymbol{\phi}_{t}\otimes\boldsymbol{\phi}_{t})$, $\mathbf{q}_{2}=\mathbf{r}_{3}+\mathbf{R}_{5}\boldsymbol{\phi}_{t}+\mathbf{R}_{6}^{H}(\boldsymbol{\phi}_{t}\otimes\boldsymbol{\phi}_{t})$. $\mathbf{Q}_{2}$ is the reshaped version of $\mathbf{q}_{2}$, i.e., $\mathbf{q}_{2}\triangleq\mathrm{vec}\{\mathbf{Q}_{2}\}$.
The above transformations have decreased the order of the optimization variable $\boldsymbol{\phi}$ from four to two. The solving difficulty has also been decreased.
However, $f_{\mathrm{obj1},\boldsymbol{\phi}}$ is still a non-convex function due to the non-convex term $\Re\{\boldsymbol{\phi}^{H}\mathbf{Q}_{2}\boldsymbol{\phi}^{\ast}\}$. To make it tractable, we continue to employ the MM technique to find a convex surrogate function. To be specific, by defining $\overline{\boldsymbol{\phi }} \triangleq [\Re \lbrace \boldsymbol{\phi}^{T}\rbrace \, \, \Im \lbrace \boldsymbol{\phi }^{T}\rbrace ]^{T}$ and $\overline{\mathbf{Q}}_{2} \triangleq\left[\begin{array}{cc}-\Re \lbrace {\mathbf{Q}_{2}}\rbrace &-\Im \lbrace {\mathbf{Q}_{2}}\rbrace 
    \\ -\Im \lbrace {\mathbf{Q}_{2}}\rbrace & \Re \lbrace {\mathbf{Q}_{2}}\rbrace \end{array}\right]$, the term $\Re\{\boldsymbol{\phi}^{H}\mathbf{Q}_{2}\boldsymbol{\phi}^{\ast}\}$ can be written as $\overline{\boldsymbol{\phi}}^{{T}}\overline{\mathbf{Q}}_{2}\overline{\boldsymbol{\phi }}$.
Then based on the second-order Taylor expansion, a convex upper bound of the non-convex term $\overline{\boldsymbol{\phi}}^{{T}}\overline{\mathbf{Q}}_{2}\overline{\boldsymbol{\phi }}$ can be derived as
\begin{equation}\begin{aligned}&\overline{\boldsymbol{\phi}}^{{T}}\overline{\mathbf{Q}}_{2}\overline{\boldsymbol{\phi }}   \\&
   \!  \leq \!\overline{\boldsymbol{\phi }}_{t}^{{\!T}}\overline{\mathbf{Q}}_{2}\overline{\boldsymbol{\phi }}_{t}\! + \!\overline{\boldsymbol{\phi }}_{t}^{T}(\overline{\mathbf{Q}}_{2}\!+\!\overline{\mathbf{Q}}_{2}^{\!T}) (\overline{\boldsymbol{\phi }}\!-\!\overline{\boldsymbol{\phi }}_{t})
    \! + \!\frac{\lambda_{1} }{2}(\overline{\boldsymbol{\phi }}\!-\!\overline{\boldsymbol{\phi }}_{t})^{\!T} \!(\overline{\boldsymbol{\phi}}\!-\!\overline{\boldsymbol{\phi }}_{t})
    \\ &=\Re \lbrace {\mathbf {q}}_{3}^ {H}\boldsymbol{\phi }\rbrace +c_{1},
   \end{aligned} \end{equation} 
 where $ {\lambda}_{1}= {\lambda}_{ {max}}(\overline{\mathbf{Q}}_{2}+\overline{\mathbf{Q}}_{2}^{T})$, ${\mathbf {q}}_{3}=\mathbf {U}(\overline{\mathbf{Q}}_{2}\!+\!\overline{\mathbf{Q}}_{2}^{\!T}-\lambda_{1} \mathbf {I}_{2 N})\overline{\boldsymbol{\phi }}_{t} $ and $c_{1}=\overline{\boldsymbol{\phi }}_{t}^{{T}}\overline{\mathbf{Q}}_{2}\overline{\boldsymbol{\phi }}_{t}-\overline{\boldsymbol{\phi }}_{t}^{T}(\overline{\mathbf{Q}}_{2}+\overline{\mathbf{Q}}_{2}^{T}) \overline{\boldsymbol{\phi }}_{t}+\frac{\lambda_{1} }{2}\overline{\boldsymbol{\phi }}_{t}^{T}\overline{\boldsymbol{\phi }}_{t}$. 
   $\mathbf {U} \triangleq[\mathbf {I}_{N} \, \jmath \mathbf {I}_{N}]$ is defined  to transform the derived real-valued function back to the complex-valued function. Until now, we have successfully found a convex surrogate function for the first six terms of the objective function \eqref{obj_phi} by exploiting the MM technique. For the remaining parts of the objective function \eqref{obj_phi}, the non-convex term is $\boldsymbol{\phi}^{T}\mathbf{R}_{9}\boldsymbol{\phi}^{\ast}$. Based on its   structure, we have
\begin{equation}\begin{aligned}
\boldsymbol{\phi}^{T}\mathbf{R}_{9}\boldsymbol{\phi}^{\ast}=\Re\{\boldsymbol{\phi}^{T}\mathbf{R}_{9}\boldsymbol{\phi}^{\ast}\}=\overline{\boldsymbol{\phi}}^{T}\overline{\mathbf{R}}_{9}\overline{\boldsymbol{\phi}},
\end{aligned}\end{equation}
where
    $\overline{\mathbf{R}}_{9} \triangleq\left[\begin{array}{cc}\Re \lbrace {\mathbf {R}}_{9}\rbrace &\Im \lbrace {\mathbf {R}}_{9}\rbrace 
    \\ -\Im \lbrace {\mathbf {R}}_{9}\rbrace & \Re \lbrace {\mathbf {R}}_{9}\rbrace \end{array}\right]$. Similar to the above procedures, by using the second-order Taylor expansion, a convex upper bound of $\overline{\boldsymbol{\phi}}^{\mathrm{T}}\overline{\mathbf{R}}_{9}\overline{\boldsymbol{\phi }} $ is derived as
\begin{equation}\begin{aligned} &\overline{\boldsymbol{\phi}}^{\mathrm{T}}\overline{\mathbf{R}}_{9}\overline{\boldsymbol{\phi }} \\& 
    \leq \overline{\boldsymbol{\phi }}_{t}^{{\!T}}\overline{\mathbf {R}}_{9}\overline{\boldsymbol{\phi }}_{t}\! + \!\overline{\boldsymbol{\phi }}_{t}^{T}(\overline{\mathbf{R}}_{9}\!\!+\!\!\overline{\mathbf{R}}_{9}^{\!T}) (\overline{\boldsymbol{\phi }}\!-\!\overline{\boldsymbol{\phi }}_{t})
    \! +\!\frac{\lambda_{2} }{2}(\overline{\boldsymbol{\phi }}\!-\!\overline{\boldsymbol{\phi }}_{t})^{\!T} \!(\overline{\boldsymbol{\phi}}\!-\!\overline{\boldsymbol{\phi }}_{t})
      \\ &=\Re \lbrace {\mathbf {q}}_{4}^{H} \boldsymbol{\phi }\rbrace +c_{2,k},
   \end{aligned} \end{equation} 
where  $ {\lambda}_{2}= {\lambda}_{ {max}}(\overline{\mathbf{R}}_{9}+\overline{\mathbf{R}}_{9}^{T})$, ${\mathbf {q}}_{4}=\mathbf {U}(\overline{\mathbf{R}}_{9}\!\!+\!\!\overline{\mathbf{R}}_{9}^{\!T}-\lambda_{2} \mathbf {I}_{2 N})\overline{\boldsymbol{\phi }}_{t} $ and $c_{2,k}=\overline{\boldsymbol{\phi }}_{t}^{{T}}\overline{\mathbf {L}}_{k}\overline{\boldsymbol{\phi }}_{t}-\overline{\boldsymbol{\phi }}_{t}^{T}(\overline{\mathbf {L}}_{k}+\overline{\mathbf {L}}_{k}^{T}) \overline{\boldsymbol{\phi }}_{t}+\frac{\lambda_{2} }{2}\overline{\boldsymbol{\phi }}_{t}^{T}\overline{\boldsymbol{\phi }}_{t}$. 
Omitting the irrelevant terms, the overall surrogate function of \eqref{obj_phi} is given by
\vspace{-2mm}
\begin{equation}\begin{aligned}
{\tilde{f}}_{\mathrm{obj},\boldsymbol{\phi}}\!&=\!-2\Re\{\mathbf{q}_{1}^{H}\boldsymbol{\phi}\}\!+\!2\Re\{\mathbf{q}_{3}^{H}\boldsymbol{\phi}\}\!\!+\!\!\frac{1}{2\rho_{1}}(2\Re\{\mathbf{r}_{8}^{H}\boldsymbol{\phi}\}\!+\!\Re \lbrace {\mathbf {q}}_{4}^ {H} \boldsymbol{\phi }\rbrace)
\\&=\Re\{\mathbf{q}_{5}^{H}\boldsymbol{\phi}\},
\end{aligned}\end{equation}
where ${\mathrm {q}}_{5}=-2\mathbf{q}_{1}+2\mathbf{q}_{3}+\frac{1}{2\rho_{1}}(2{\mathbf {r}}_{8}+\mathbf{q}_{4})$.
Then, the optimization w.r.t. $\boldsymbol{\phi }$ is reformulated as 
\begin{equation}\text {(P10)} \quad \min _{  \boldsymbol{\phi} }~ \Re\{\mathbf{q}_{5}^{H}\boldsymbol{\phi}\}
 {\qquad  ~\text {s.t.}~} 
    |\phi  _{{i}}|=1, \forall i.\label{23c}
      \end{equation}
Based on the phase alignment, the optimal closed-form solution of problem (P10) is derived as
\begin{equation}\begin{aligned}
\boldsymbol{\phi}^{\star}=-e^{ j\text {arg}(\mathbf {q}_{5})}.
\label{opt_phi}\end{aligned}\end{equation}

\subsubsection{Update \texorpdfstring{$\mathbf{a}$}{\texttwoinferior}}  With fixed $\boldsymbol{\phi}$, $s_{k,i}$ and $\mathbf{a}_{a}$, the sub-problem w.r.t. $\mathbf{a}$ is formulated as
\begin{subequations}\begin{align}\text {(P11)} &\quad \min _{  \mathbf{a} }~ f_{\mathrm{obj},\mathbf{a}} \label{obj_a}
\\& {\qquad  ~\text {s.t.}~} 
   {\mathbf{a}}(i)\in\{0,1\}, \quad i\in \{1,\cdots,N\}.\label{0-1}
     \end{align}  \end{subequations} 
It is obvious that problem (P11) is a binary integer optimization problem, which is NP-hard. However, as mentioned before, the optimization variables $\boldsymbol{\phi}$ and $\mathbf{a}$ have many similarities. By checking problem (P11), we further find that the objective function \eqref{obj_a} has the similar form with \eqref{obj_phi}. They are both of high orders and both involve Kronecker product operation. Moreover, based on the 0-1 constraint \eqref{0-1}, we have $\mathbf{a}^{T}\mathbf{a}=a$, which is a constant. The above discussions motivate us to also employ MM technique to tackle this challenging problem. Specifically, we define $\tilde{\mathbf{a}}_{1}=\begin{bmatrix}\mathbf{a}^{T}  & (\mathbf{a}\!\otimes\!\mathbf{a})^{T}
\end{bmatrix}^{T}$, $\tilde{\mathbf{R}}_{2}=\begin{bmatrix}\mathbf{R}_{13}&\mathbf{R}_{14} ^{H} \\ \mathbf{R}_{14}&\mathbf{R}_{15}
\end{bmatrix}$. As expected, $\tilde{\mathbf{R}}_{2}$ is a positive semi-definite matrix. Then based on the first-order Taylor expansion, we have
\begin{equation}\begin{aligned}
-\tilde{\mathbf{a}}_{1}^{T}\tilde{\mathbf{R}}_{2}\tilde{\mathbf{a}}_{1}\leq \tilde{\mathbf{a}}_{1t}^{T}\tilde{\mathbf{R}}_{2}\tilde{\mathbf{a}}_{1t}-2\Re\{\tilde{\mathbf{a}}_{1t}^{T}\tilde{\mathbf{R}}_{2}\tilde{\mathbf{a}}_{1}\}.
\end{aligned}\end{equation}
Omitting the irrelevant terms, combining the term $-2\Re\{\tilde{\mathbf{a}}_{1t}^{T}\tilde{\mathbf{R}}_{2}\tilde{\mathbf{a}}_{1}\}$ with the terms $-2\Re\{\mathbf{r}_{11}^{H}\mathbf{a}\}$ and $-2\Re\{\mathbf{r}_{12}^{H}(\mathbf{a}\otimes\mathbf{a})\}$, we have
\begin{equation}\begin{aligned}
&-2\Re\{\mathbf{r}_{11}^{H}\mathbf{a}\}
-2\Re\{\mathbf{r}_{12}^{H}(\mathbf{a}\otimes\mathbf{a})\}
-2\Re\{\tilde{\mathbf{a}}_{1t}^{T}\tilde{\mathbf{R}}_{2}\tilde{\mathbf{a}}_{1}\}
\\&=-2\Re\{\mathbf{q}_{6}^{H}\mathbf{a}+\mathbf{q}_{7}^{H}(\mathbf{a}\!\otimes\!\mathbf{a})\}
\\&=-2\Re\{\mathbf{q}_{6}^{H}\mathbf{a}+\mathbf{a}^{T}\mathbf{Q}_{7}\mathbf{a}\},
\end{aligned}\end{equation}
where $\mathbf{q}_{6}=\mathbf{r}_{11}+\mathbf{R}_{13}^{H}{\mathbf{a}}_{t}+\mathbf{R}_{14}^{H}(\mathbf{a}_{t}\otimes\mathbf{a}_{t})$, $\mathbf{q}_{7}=\mathbf{r}_{12}+\mathbf{R}_{14}{\mathbf{a}}_{t}+\mathbf{R}_{15}^{H}(\mathbf{a}_{t}\!\otimes\!\mathbf{a}_{t})$, 
${\mathbf{Q}}_{7}$ is the reshaped version of ${\mathbf{q}}_{7}$, i.e., ${\mathbf{q}}_{7}=\mathrm{vec}\{{\mathbf{Q}}_{7}\}$.
Since $\mathbf{a}^{T}\mathbf{Q}_{7}\mathbf{a}$ is still non-convex, we continue to employ MM technique to find a convex surrogate function. Defining $\tilde{\mathbf{Q}}_{7}=\Re\{-\mathbf{Q}_{7}\}$, we have $-\Re\{\mathbf{a}^{T}\mathbf{Q}_{7}\mathbf{a}\}=\mathbf{a}^{T}\tilde{\mathbf{Q}}_{7}\mathbf{a}$. Then based on the second-order Taylor expansion, a convex surrogate function of $\mathbf{a}^{T}\tilde{\mathbf{Q}}_{7}\mathbf{a}$ can be derived as 
\begin{equation}\begin{aligned}&
\mathbf{a}^{T}\tilde{\mathbf{Q}}_{7}\mathbf{a}\\&\leq\mathbf{a}_{t}^{T}\tilde{\mathbf{Q}}_{7}\mathbf{a}_{t} \!+\!\mathbf{a}_{t}^{T}(\tilde{\mathbf{Q}}_{7}\!+\!\tilde{\mathbf{Q}}_{7}^{T})(\mathbf{a}-\mathbf{a}_{t})\!+\!\frac{\lambda_{3}}{2}(\mathbf{a}-\mathbf{a}_{t})^{T}(\mathbf{a}-\mathbf{a}_{t}),
\end{aligned}\end{equation}
where  $ {\lambda}_{3}= {\lambda}_{ {max}}(\tilde{\mathbf{Q}}_{7}+\tilde{\mathbf{Q}}_{7}^\mathrm{T})$. 
Omitting the known terms irrelevant with $\mathbf{a}$, we can find that the remaining term is linear. Then, combining it with $-2\Re\{\mathbf{q}_{6}^{H}\mathbf{a}\}$, we have
\begin{equation}\begin{aligned}
-2\Re\{\mathbf{q}_{6}^{H}\mathbf{a}\}+2\Re\{\mathbf{a}_{t}^{T}(\tilde{\mathbf{Q}}_{7}+\tilde{\mathbf{Q}}_{7}^{T})\mathbf{a}-\lambda_{3}\mathbf{a}_{t}^{T}\mathbf{a}\}=\Re\{\mathbf{q}_{8}^{H}\mathbf{a}\},
\end{aligned}\end{equation}
where $\mathbf{q}_{8}=-2\mathbf{q}_{6}+2(\tilde{\mathbf{Q}}_{7}+\tilde{\mathbf{Q}}_{7}^{T}-\lambda_{3}\mathbf{I}_{N})\mathbf{a}_{t}$.
For the term $\mathbf{a}^{T}\mathbf{R}_{18}\mathbf{a}$, since $\mathbf{R}_{18}$ is a Hermitian matrix, we can use the second-order Taylor expansion to find a linear surrogate function.
\begin{equation}\begin{aligned}&
\mathbf{a}^{T}\mathbf{R}_{18}\mathbf{a}\\&\leq\mathbf{a}^{T}{\Lambda}_{4}\mathbf{a} +2\Re\{\mathbf{a}^{T}(\mathbf{R}_{18}-{\Lambda}_{4})\mathbf{a}_{t}\}+\mathbf{a}_{t}^{T}({\Lambda}_{4}-{\mathbf{R}}_{18})\mathbf{a}_{t},
\end{aligned}\end{equation}
where $\boldsymbol {\Lambda}_{4}\!=\! {\lambda}_{ {max}}({\mathbf{R}}_{18}){\mathbf{I}}_{N}$. Combining the term $2\Re\{\mathbf{a}^{T}(\mathbf{R}_{18}\!-\!{\Lambda}_{4})\mathbf{a}_{t}\}$ with $2\Re\{\mathbf{r}_{17}^{H}\mathbf{a}\}$, we have
\begin{equation}\begin{aligned}
\frac{1}{2\rho_{1}}(2\Re\{\mathbf{r}_{17}^{H}\mathbf{a}+\mathbf{a}^{T}(\mathbf{R}_{18}-{\Lambda}_{4})\mathbf{a}_{t}\})=\frac{1}{\rho_{1}}\Re\{\mathbf{q}_{9}^{H}\mathbf{a}\},
\end{aligned}\end{equation}
where $\mathbf{q}_{9}=\mathbf{r}_{17}+(\mathbf{R}_{18}-{\Lambda}_{4})\mathbf{a}_{t}$. By adding all the linear surrogate functions together, the objective function is simplified as 
\begin{equation}\begin{aligned}
\tilde{f}_{\mathrm{obj},\mathbf{a}}=\Re\{\mathbf{q}_{8}^{H}\mathbf{a}\}
+ \frac{1}{\rho_{1}}\Re\{\mathbf{q}_{9}^{H}\mathbf{a}\}
+\frac{1}{2\rho_{2}}(\mathbf{r}_{19}^{T}\mathbf{a})
=\Re\{\mathbf{q}_{10}^{H}\mathbf{a}\}
,\end{aligned}\end{equation}
where $\mathbf{q}_{10}=\mathbf{q}_{8}+\frac{1}{\rho_{1}}\mathbf{q}_{9}+\frac{1}{2\rho_{2}}\mathbf{r}_{19}$.
Then, the optimization problem w.r.t. $\mathbf{a}$ is expressed as 
\begin{subequations}\begin{align}\text {(P12)} &\quad \min _{  \mathbf{a} }~ \Re\{\mathbf{q}_{10}^{H}\mathbf{a}\}
\\& {\qquad  ~\text {s.t.}~} 
     {\mathbf{a}}(i)\in\{0,1\}, \quad i\in \{1,\cdots,N\}.
     \end{align}  \end{subequations}
If we denote the set of first $a$ minimum elements of $\Re\{\mathbf{q}_{10}\}$ as $\mathbb{L}$, the optimal $\mathbf{a}^{\star}$ can be obtained as
\begin{equation}
\mathbf{a}^{\star}(i)=
\left\{
\begin{aligned}
1,& \quad\Re\{\mathbf{q}_{10}\}(i) \in\mathbb{L} \\\
0,& \quad\Re\{\mathbf{q}_{10}\}(i) \notin\mathbb{L}.
\end{aligned}
\right.
\label{opt_a}\end{equation}
The optimal $\mathbf{a}^{\star}$ implies that we need to find first $a$ minimum elements from $\Re\{\mathbf{q}_{10}\}$. The corresponding indexes are the locations of the elements working on the connected mode.

\subsubsection{Update \texorpdfstring{$\mathbf{a}_{a}$}{\texttwoinferior}} With fixed $\boldsymbol{\phi}$, $s_{k,i}$ and $\mathbf{a}$, the sub-problem w.r.t. $\mathbf{a}_{a}$ is formulated as
\begin{subequations}\begin{align}\text {(P13)} &\quad \min _{  \mathbf{a}_{a}}~ f_{\mathrm{obj},\mathbf{a}_{a}}
\\& {\qquad  ~\text {s.t.}~} 
     \sum _{l=1}^{N} \mathbf{a}_{ai} [l]=1, \mathbf{a}_{ai} [l]\in\{0,1\}, \forall i, \forall l.
     \end{align}  \end{subequations}
This problem is similar to problem (P11). Following the same procedure,
 the optimal solution $\mathbf{a}_{a}^{\star}$ can be obtained. Therefore, the detailed derivations are omitted here. 

    \begin{algorithm}[t]
    \caption{Joint beamforming and mode selection for RDARS-aided ISAC.}
    \label{alg:example}
    \begin{algorithmic}[1] 
    \REQUIRE ~~$\mathbf {h}_{bu,k}$, $\mathbf {h}_{ru,k}$, $\mathbf {H}_{br}$, $\mathbf {h}_{bt}$, $\mathbf {h}_{rt}$, $P$, $\sigma _{1,k}^{2}$, $\sigma _{2}^{2}$, $\sigma _{\alpha}^{2}$, $\overline{\gamma}_{k}$, $\rho_{1}$, $\rho_{2}$.
    \STATE {Initialize: ${\boldsymbol {\Phi} }^{0} $, $ {\mathbf{F}}^{0}  $, $\mathbf {A}^{0}$, $\mathbf{A}_{a}^{0}$}.  
    \REPEAT 
    \STATE {Update $\mathbf{w}$  by solving problem (P2).}
    \STATE {Update   ${\mathbf{f}}$  by solving problem (P4).}
    \STATE {Update $\boldsymbol{\phi}$ based on \eqref{opt_phi}.}
    \STATE {Update $s_{k,i}$ based on \eqref{opt_ski} .}
    \STATE {Update $\mathbf{A}$ based on \eqref{opt_a}.}
    \STATE {Update $\mathbf{A}_{a}$ by solving problem (P13).}
    \STATE {$\rho_{1}:=c_{1}\rho_{1}$, $\rho_{2}:=c_{2}\rho_{1}$}
    \UNTIL {the convergence is satisfied.}
    \ENSURE ~~ $\mathbf{w}^{\star}$, $\mathbf{f}^{\star}$, $\boldsymbol {\Phi}^{\star} $, $\mathbf{A}^{\star}$,  $\mathbf{A}_{a}^{\star}$.  
    \end{algorithmic}
    \end{algorithm}

\vspace{-2mm}
\subsection{Convergence and Complexity Analysis}
With the derivations of the optimal receive filter $\mathbf{w}^{\star}$, the  optimal transmit beamforming matrix $\mathbf{F}^{\star}$, the  optimal RDARS reflection matrix $\boldsymbol{\Phi}^{\star}$ and the optimal selection matrices $\mathbf{A}^{\star}$ and $\mathbf{A}_{a}^{\star}$, the overall optimization process is summarized in Algorithm~\ref{alg:example}. 
This iterative process ensures that the objective value of problem (P1) continually improves in each iteration. Moreover, because of the communication SINR constraint, this objective value is upper-bounded by a finite limit. Consequently, the proposed algorithm is guaranteed to converge to a locally optimal solution \cite{7558213}.

Next, we discuss the computational complexity of the proposed algorithm. In Algorithm~\ref{alg:example}, the original optimization problem is decoupled into several sub-problems. For the first sub-problem, where we optimize the receive filter $\mathbf{w}$, the complexity lies in its eigenvalue decomposition (EVD), which is given by $\mathcal {O} \big (  {M}^{3} \big)$. For the transmit beamforming matrix $\mathbf{F}$ optimization, it's a SOCP problem whose computation complexity is given by $\mathcal {O} \big (((M+a)K)^{3.5} \big)$. The complexities for optimizing $\boldsymbol {\Phi}$, $\mathbf{A}$ and  $\mathbf{A}_{a}$ are all dominated by the EVD, which are given by  $\mathcal {O} \big (  ({2N})^{3} \big)$,  $\mathcal {O} \big (  ({2N})^{3} \big)$ and  $\mathcal {O} \big (  {N}^{3} \big)$, respectively. The complexity for optimizing $\boldsymbol {\Phi}$ is given by $\mathcal {O} \big (   K\log_2(\frac{1}{\epsilon})N_{r}^{2} \big)$, where $\epsilon$ is the iteration accuracy. Therefore, denoting the total number of required iterations as $I$, the overall computational  complexity of Algorithm~\ref{alg:example} can be calculated as $\mathcal {O} \big (I( {M}^{3}+((M+a)K)^{3.5}+K\log_2(\frac{1}{\epsilon})N_{r}^{2}+{17N}^{3} )   \big)$. 


\renewcommand{\algorithmicrequire}{ \textbf{Input:}} 
\renewcommand{\algorithmicensure}{ \textbf{Output :}} 

\vspace{-2mm}
\section{Numerical and Simulation Results}
\begin{table}[t] 
\setlength{\tabcolsep}{3pt} 
\renewcommand{\arraystretch}{1.2} 
\begin{center}  
\caption{The Main Simulation Parameters.}  
\label{table1}
\begin{tabular}{c c c}   
\hline   \textbf{Description} & \textbf{Parameter} & \textbf{Value} \\   
\hline   
Number of transmit/receive antennas at BS & M& 16  \\ 
Number of RDARS elements & N & 120 \\  
Number of connected elements  & a & 3  \\ 
Total transmit power  & P & 20 dBm  \\  
Noise power  & \(\sigma_{1,k},\sigma_{2} \) & -80 dBm  \\ 
AoD towards RDARS at BS  & \(\theta_{br}^{D}\)& \(\frac{\pi}{2}\)  \\ 
Azimuth AoA at RDARS   & \(\theta_{br}^{A}\) &\(\frac{\pi}{4}\) \\ 
Azimuth AoD towards target at RDARS  &\(\theta_{rt}^{D}\)  & \(\frac{\pi}{4}\)  \\ 
Azimuth AoD towards target at BS  & \(\theta_{bt}^{D}\)  &\(\mathrm{arctan}(\frac{d_{br}}{d_{rt}})\) \\ 
Elevation AoD towards target at RDARS   &\(\psi_{rt}^{D}\)  &\(\mathrm{arctan}(\frac{d_{H}}{d_{rt}})\)  \\ 
Path loss exponent for channel $\mathbf{H}_{br}$ & \({\alpha }_{{H}_{br}}\)&2.4 \\ 
Path loss exponent for channel $\mathbf{h}_{bt}$  & \({\alpha }_{{h}_{bt}}\)& 2.3  \\ 
Path loss exponent for channel $\mathbf{h}_{bu,k}$   & \({\alpha }_{{h}_{bu,k}}\)& 3.0  \\ 
Path loss exponent for channel $\mathbf{h}_{rt}$ &\({\alpha }_{{h}_{rt}}\)& 2.0   \\ 
Path loss exponent for channel $\mathbf{h}_{ru,k}$ & \({\alpha }_{{h}_{ru,k}}\)& 2.6   \\ 
\hline   
\end{tabular}   
\end{center}   
\end{table}


This section evaluates the proposed RDARS-aided ISAC framework and our joint beamforming and mode selection design via numerical and simulation results.
We consider a 3D Cartesian coordinate system, and assume that the BS, RDARS and the target are located in (15, 0, 5)m, (0, 0, 5)m, and (0, 10, 0)m, respectively.
The users are randomly distributed within a circle centered at (50, 50, 0)m with a radius of 5m.

The main simulation parameters are summarized in Table \ref{table1}.
The azimuth AoD at the BS towards the RDARS is set to $\theta_{br}^{D}=\frac{\pi}{2}$.
The azimuth AoA at the RDARS is set to
$\theta_{br}^{A}=\frac{\pi}{4}$.
The azimuth AoD at the RDARS towards the target  is set to $\theta_{rt}^{D}=\frac{\pi}{4}$.
We denote the azimuth distance between the BS and RDARS by $d_{br}$, and that between the RDARS and the target by $d_{rt}$.
Let the height of the BS be $d_{H}$.
The azimuth AoD at the BS towards the target and the elevation AoD at the RDARS  towards the target can be written as
$$\theta_{bt}^{D}=\mathrm{arctan}\left(\frac{d_{br}}{d_{rt}}\right),~~~\psi_{rt}^{D}=\mathrm{arctan}\left(\frac{d_{H}}{d_{rt}}\right).$$

We adopt the distance-based path loss model $P\!=\!P_{0}(d/d_{0})^{-\alpha}$, where $P_{0}=-30$ dB is the reference path loss at $d_{0}=1$m. 
We set the path loss exponents for the links between BS and RDARS, BS and target, and BS and user $k$ as ${\alpha }_{{H}_{br}}=2.4,{\alpha}_{{h}_{bt}}=2.3$ and ${\alpha }_{{h}_{bu,k}}=3.0$, respectively.
On the other hand, the path loss exponents for the link between the RDARS and target,  and that between RDARS and user $k$ are set to ${\alpha }_{{h}_{rt}}=2.0$ and ${\alpha }_{{h}_{ru,k}}=2.6$, respectively.

The communication SINR threshold for each user is set to $\overline{\gamma}_{k}=10$ dB. Unless otherwise specified, we set the noise power to $\sigma_{1,k}=\sigma_{2}=-80$ dBm and the total transmit power to $P=20$ dBm.
The BS is equipped with $M_t=16$ transmit antennas and $M_r=16$ receive antennas. The RDARS has 120 elements in total, among which three elements working in the connected mode.
The initial values of the penalty coefficients are set to $\rho_{1}=10^{3}$ and $\rho_{2}=10^{5}$, respectively.


\begin{figure}[t]
    \centering
    \includegraphics[width=3.1in]{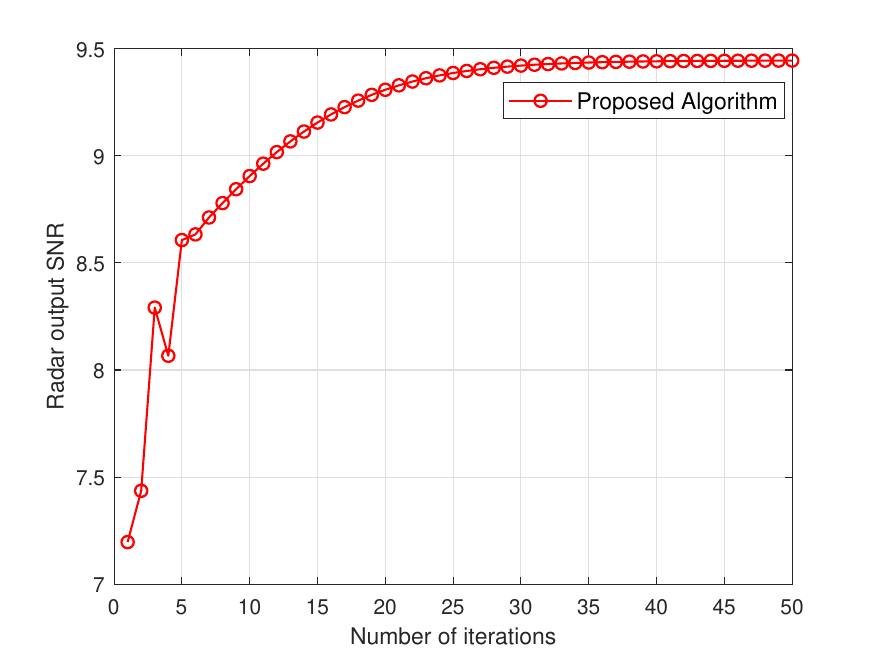}
   \vspace{-3mm}
    \caption{Convergence behavior of the proposed joint beamforming and mode selection for RDARS-aided ISAC. }
   \label{fig_convergence}
\end{figure}

To demonstrate the superiority of our RDARS-aided ISAC framework and the jointly designed beamformer and mode selection, we have established a set of benchmarks for evaluation. To enhance clarity in our discussion, we have assigned specific names to each of these schemes based on their reflective surface configuration and primary function. For example, our scheme adopts RDRAS as the transmitting and reflective surface to achieve ISAC, thus it is denoted by ``RDARS-ISAC''.
\begin{enumerate}
\item
\textbf{RDRAS-sensing, optimized phase:} In this scenario, RDRAS is used exclusively for sensing, and the phase shifts are optimized using our approach.
\item
\textbf{RDRAS-sensing, random phase:} RDRAS is used only for sensing, and we assume that the phase for each RDARS reflection coefficient is independently and uniformly sampled from [0, 2$\pi$). 
\item
\textbf{DAS-ISAC:} The RDARS has no reflective elements, but only three active elements working in the connected mode. That is, this is a DAS and the goal is to perform ISAC. 
\item
\textbf{DAS-sensing:} The RDARS has three active elements only, and the system is designed for sensing.
\item
\textbf{RDARS-ISAC, random phase:} The proposed RDARS-ISAC framework, but the phase shifts are not optimized.
\item
\textbf{Passive RIS-ISAC:} The passive RIS-aided ISAC system proposed in \cite{luo2023ris}.
\end{enumerate}






It is worth noting that our RDARS-aided system is general and can cover the above systems as special cases by configuring the RDARS elements at different working modes accordingly. For example, by setting the selection matrix $\mathbf{A}$ as all-zero matrix, the RDARS-aided system reduces to the passive RIS system. By setting the phase shift matrix as all-zero matrix, the RDARS-aided system reduces to the DAS system. Under sensing-only scenario, it can also be regarded as a distributed MIMO radar system. The design freedom on the selection matrix and the phase shifts offers promising flexibility and substantial gain to support a wide range of practical applications with various performance requirements.

To start with, we verify the convergence behavior of the proposed joint design for RDARS-aided ISAC in Algorithm~\ref{alg:example}. As can be seen from Fig. \ref{fig_convergence}, with the increase in the number of iterations, the radar output SNR continues to increase and  finally converges to a fixed value, demonstrating the robust convergence of the proposed algorithm.

\begin{figure}[t]
    \centering
    \includegraphics[width=3.1in]{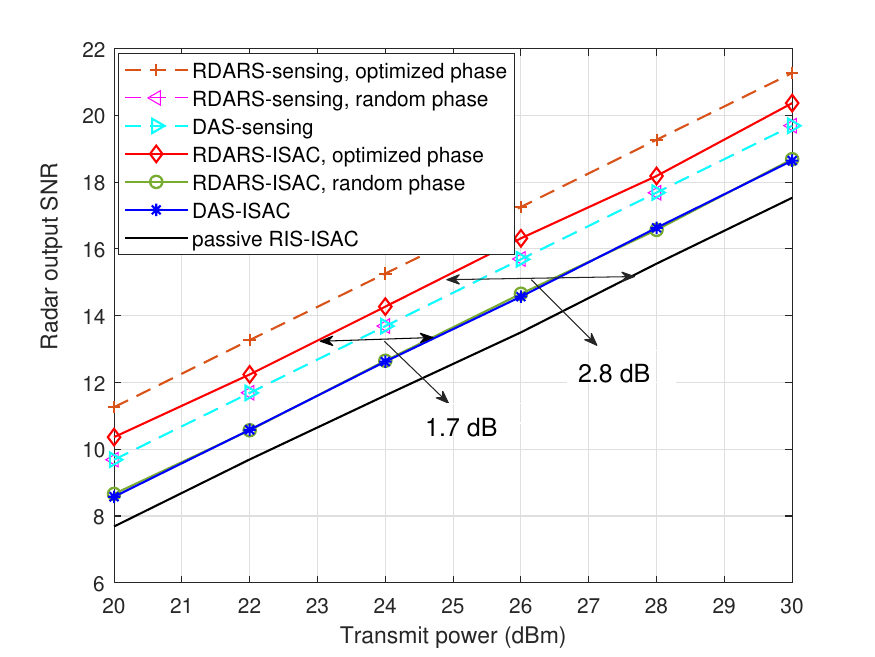}
    \vspace{-2mm}
    \caption{Radar output SNR versus transmit power.} 
   \label{fig_B}
    \vspace{-2mm}
\end{figure}  
 
\begin{figure}[t]
    \centering
    \includegraphics[width=3.1in]{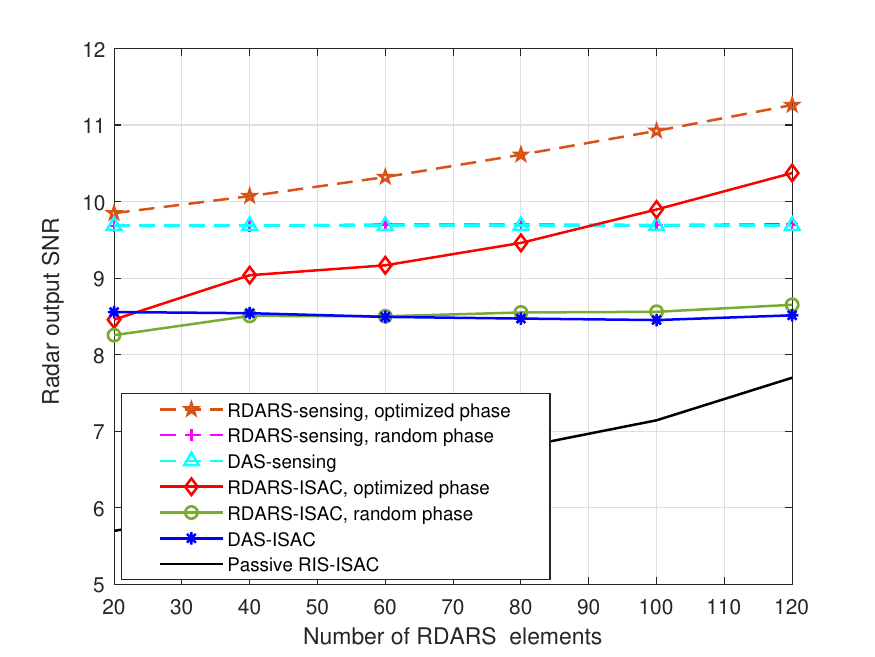}
    \vspace{-2mm}
    \caption{Radar output SNR versus the number of total RDARS elements.}
   \label{fig_lr}
    \vspace{-2mm}
\end{figure} 
 
\begin{figure}[t]
    \centering
    \includegraphics[width=3.1in]{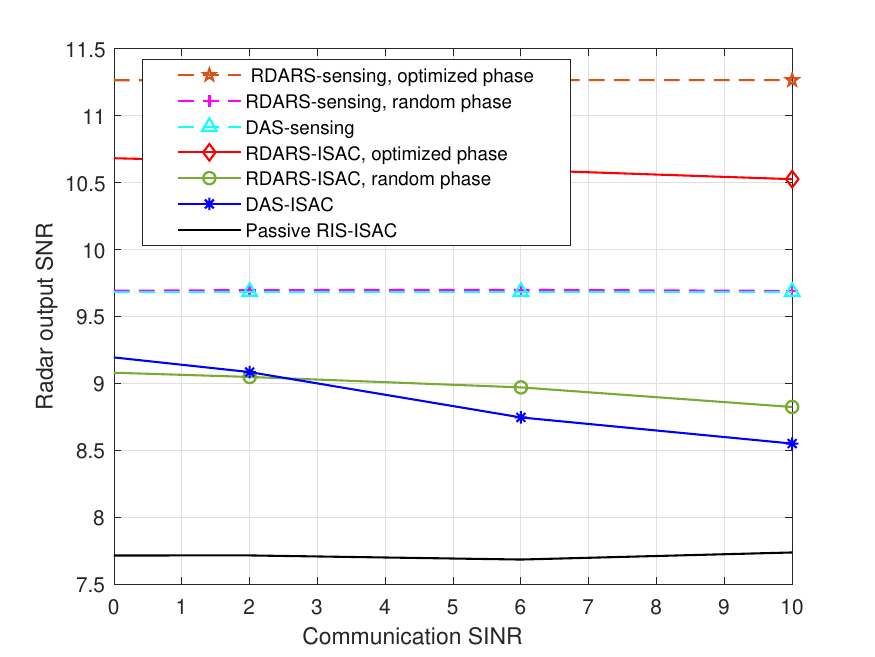}
    \vspace{-2mm}
    \caption{Radar output SNR versus communication SINR.}
   \label{fig_communication}
    \vspace{-2mm}
    \end{figure} 


\begin{figure}[t]
    \centering
    \includegraphics[width=3.1in]{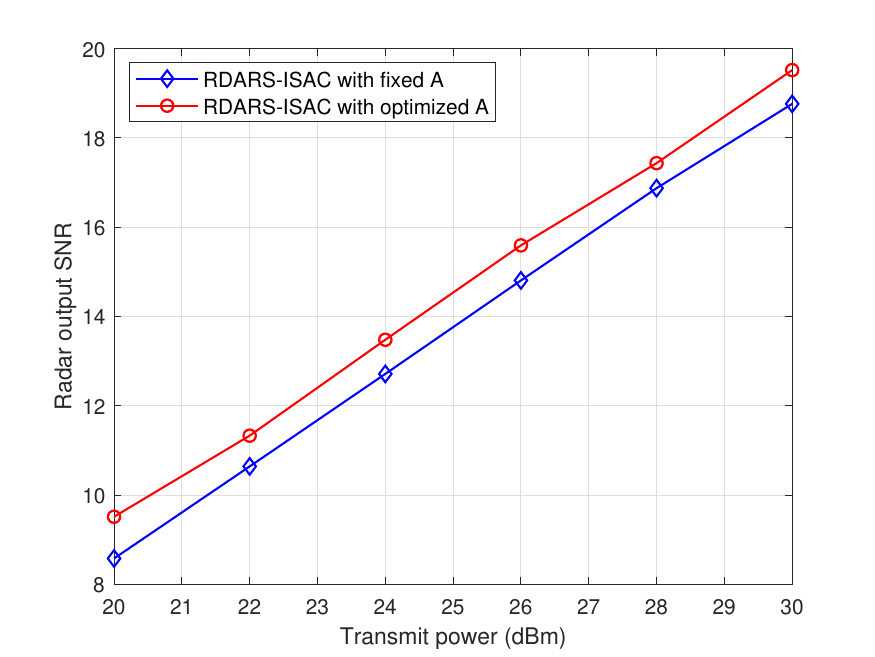}
    \vspace{-2mm}
    \caption{Radar output SNR versus transmit power for RDARS-ISAC system with fixed and optimized selection matrix $\mathbf{A}$.} 
   \label{fig_A}
\end{figure}
Then, we evaluate our RDARS-ISAC system benchmarked against the above baselines.
Fig. \ref{fig_B} presents the radar output SNR versus the transmit power.
We have four main observations:
\begin{itemize}
    \item As the transmit power increases, all schemes achieve better performance. Passive RIS-ISAC performs the worst as the passive RIS only provides the passive beamforming gain. Compared with the passive RIS-ISAC, the RDARS-ISAC can provide about 2.8dB performance enhancement, which shows that even a few connected elements of RDARS can significantly boost system performance enhancement.
    \item The ``sensing'' systems perform better than the corresponding ``ISAC'' systems. This matches our intuitive as all the power in ``sensing'' systems can be used for radar sensing, while ``ISAC'' systems have to accommodate both sensing and communication performances.
    \item RDARS-ISAC with optimized phase shifts outperforms RDARS-ISAC with random phase shifts. This confirms the effectiveness of the proposed algorithm.
    \item RDARS-ISAC can achieve about 1.7dB performance improvement in contrast to DAS-ISAC, thanks to the passive beamforming gain brought by RDARS.
\end{itemize}

From these observations, we can conclude that by incorporating not only connected elements as distributed antennas but also the reflecting elements into the RDARS, substantial distribution and reflection gains can be achieved by our proposed optimal design.

Fig. \ref{fig_lr}  presents the radar output SNR versus the number of total RDARS elements. As the number of RDARS elements increases, the RDARS-sensing with optimized phase, RDARS-ISAC with optimized phase and  passive RIS-RDARS schemes all exhibit better performance since more passive reflecting elements are included and then bring higher reflection gain. For DAS-sensing and DAS-ISAC systems, since they don't have passive reflecting elements, the performance remains unchanged as we increase the number of the total RDARS elements while keeping the number of connected elements as $a=3$. Though the RDARS-sensing with random phase and RDARS-ISAC with random phase systems equip with passive reflecting elements, these elements adopt the random phase, which can't provide passive beamforming gains. Therefore, the RDARS-sensing with random phase and RDARS-ISAC with random phase schemes nearly have the same performance with DAS-sensing  and DAS-ISAC  schemes, respectively.

Fig. \ref{fig_communication} presents the radar output SNR versus the communication SINR for all schemes. 
A first observation is that the performance of all ISAC systems degrades when a higher communication SINR threshold is set. This is intuitive as more power is needed for communication system to satisfy the requirement; hence, there is a performance trade off between  sensing and communication. On the other hand, the performance of  ``sensing'' systems remains unchanged with the increase in the communication SINR requirement.


Fig. \ref{fig_A} verifies the effectiveness of the proposed algorithm for optimizing the selection matrix $\mathbf{A}$. In ``RDARS-ISAC with fixed $\mathbf{A}$'', the locations of the elements working in the connected mode are predefined and fixed. Specifically, we assume that the first three diagonal  elements of the RDARS work in the connected mode.  By contrast, in ``RDARS-ISAC with optimized $\mathbf{A}$'', the  locations of the elements working on the connected mode are optimized based on  Algorithm~\ref{alg:example}. Fig. \ref{fig_A} confirms that ``RDARS-ISAC with optimized $\mathbf{A}$'' exhibits higher performance gains than ``RDARS-ISAC with fixed $\mathbf{A}$''. This gain is the unique selection gain brought by the RDARS architecture with flexible mode configuration. Thus, optimizing the locations of the elements working in the connected mode is of great significance for RDARS-ISAC system, and our proposed algorithm has proven its effectiveness in this regard.

\begin{figure*}[t]
	\centering
	\begin{minipage}{0.35\linewidth}
		\centerline{\includegraphics[width=1.1\textwidth]{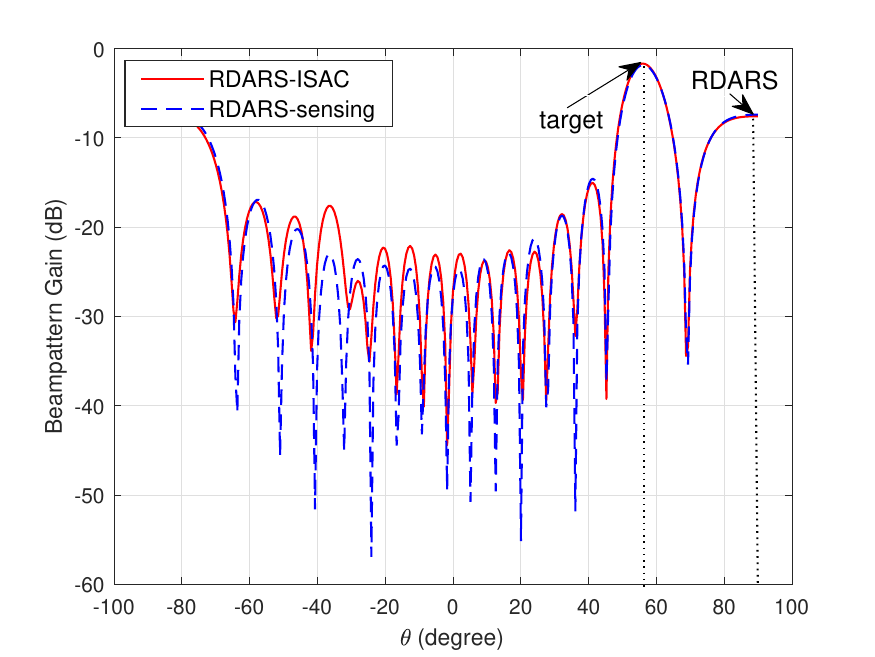}}
		\centerline{(a)}
	\end{minipage}\hspace{1cm}
	\begin{minipage}{0.35\linewidth}
		\centerline{\includegraphics[width=1.1\textwidth]{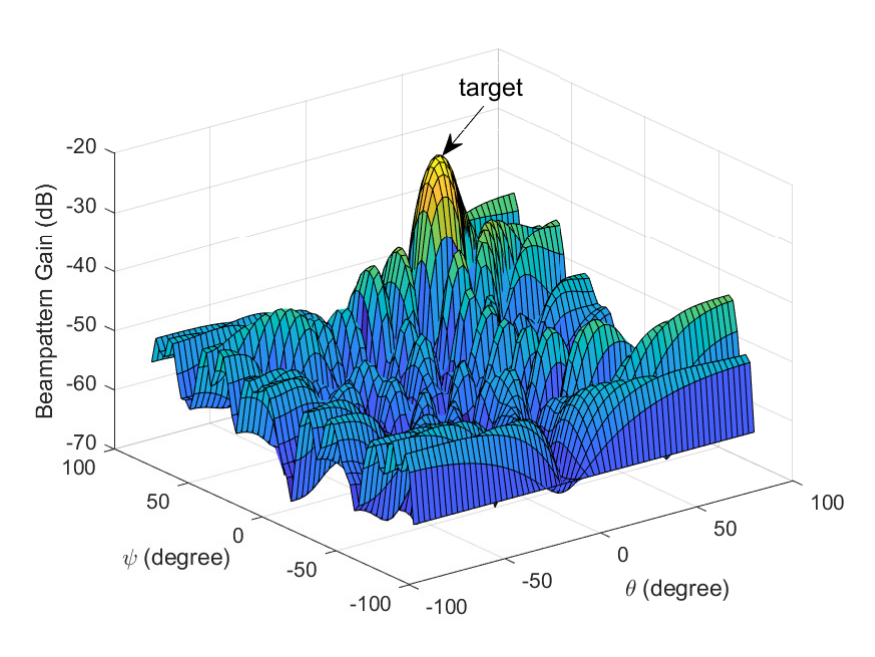}}	 
		\centerline{(b)}
	\end{minipage}
	\caption{(a) The achieved beampattern gain at the BS versus the AoD $\theta$; (b) The achieved beampattern gain at the RDARS versus the azimuth AoD $\theta$ and elevation AoD $\psi$.}
	\label{gain}
\end{figure*}

Finally, we evaluate the achieved beampattern gain at the BS versus the AoD $\theta$, and that at the RDARS versus the azimuth AoD $\theta$ and elevation AoD $\psi$.
The simulation results are illustrated in Fig. \ref{gain}(a) and Fig. \ref{gain}(b), respectively.
From Fig. \ref{gain}(a), it can be observed that the proposed transmit beamformers can strongly direct the mainlobes toward the directions where the target and the RDARS are located. In addition, as shown in Fig. \ref{gain}(b),  a strong beam can also be generated at RDARS to direct the signals to the target, thanks to the passive beamforming optimization for the reflection elements of RDARS.
Fig. \ref{gain}(a) also shows that the sidelobes of the ISAC system is higher than those of the only radar system, which is due to the additional SINR constraints for the communication system.


 \vspace{-4mm}
\section{Conclusion}

This paper introduced a new RDARS-aided ISAC framework and provided valuable insights into the reflection, distribution and selection gains it offers in radar and communication systems.
The flexible functionality of RDARS, operating as both reflection devices and transmit antennas, exemplifies its transformative potential in elevating the capabilities of ISAC systems.
Furthermore, the presented joint optimization algorithm serves as a promising approach to tackle the intricate optimization challenges that such systems often face. Its application not only enhances the performance of RDARS-aided ISAC but also underscores the feasibility of this technology in practical scenarios.

Moving forward, our research will continue to explore various aspects of RDARS-aided ISAC, including real-world constraint  and practical considerations, thus paving the way for the widespread adoption of this new framework in future wireless communication and radar applications.

\vspace{-3mm}
 \begin{appendices}
\section{ Proof of Lemma 1 }\label{sec:AppA}
\vspace{-1mm}
We first express the objective function w.r.t. $\boldsymbol{\phi}$ explicitly. According to the diagonal structure of $\boldsymbol{\Phi}$, we have $\boldsymbol{\Phi} \mathbf{h}_{rt}=\mathrm{diag}(\mathbf{h}_{rt})\boldsymbol{\phi}$. 
 Defining $\mathbf{H}_{3}\triangleq(\mathbf{I}-\mathbf{A})\mathbf{H}_{br}$, $\mathbf{h}_{4}\triangleq\mathbf{h}_{bt}+\mathbf{H}_{3}\boldsymbol{\Phi} \mathbf{h}_{rt}$, we have
\vspace{-2mm}
\begin{align}
&\mathbf{w}^{H}{\mathbf{H}}_{2}{\mathbf{f}}_{k}\\&=\mathbf{w}^{H}\mathbf{h}_{4}\begin{bmatrix}\mathbf{h}_{4}^{T}
&\mathbf{h}_{rt}^{{T}}\mathbf{A}_{a}\end{bmatrix}\begin{bmatrix}\mathbf{f}_{k1}\\ \mathbf{f}_{k2}\end{bmatrix}
\\&=\mathbf{w}^{H}\mathbf{h}_{bt}\mathbf{h}_{bt}^{{T}}\mathbf{f}_{k1}+\mathbf{w}^{H}\mathbf{h}_{bt}\mathbf{h}_{rt}^{{T}}\mathbf{A}_{a}\mathbf{f}_{k2}\\& \quad+[(\mathrm{diag}(\mathbf{h}_{rt})\mathbf{H}_{3}\mathbf{f}_{k1})^{T}
\otimes\mathbf{w}^{H}\mathbf{h}_{bt}
\\&\quad+
(\mathbf{h}_{bt}^{{T}}\mathbf{f}_{k1})^{T}\otimes\mathbf{w}^{H}\mathbf{H}_{3}^{T}\mathrm{diag}(\mathbf{h}_{rt})
\\&\quad+(\mathbf{h}_{rt}^{{T}}\mathbf{A}_{a}\mathbf{f}_{k2})^{T}
\otimes\mathbf{w}^{H}\mathbf{H}_{3}^{T}\mathrm{diag}(\mathbf{h}_{rt})]\mathrm{vec}(\boldsymbol{\phi})
\\&\quad+[(\mathrm{diag}(\mathbf{h}_{rt}))\mathbf{H}_{3}\mathbf{f}_{k1} )^{T}\otimes\mathbf{w}^{H}\mathbf{H}_{3}^{T}\mathrm{diag}(\mathbf{h}_{rt})]\mathrm{vec}(\boldsymbol{\phi}\boldsymbol{\phi}^{T})
\\&={x}_{1,k}+\mathbf{Y}_{1,k}\boldsymbol{\phi}+\mathbf{Z}_{1,k}(\boldsymbol{\phi}\otimes\boldsymbol{\phi}),
\end{align}
where we define
\begin{subequations}\begin{align}
&{x}_{1,k}\triangleq\mathbf{w}^{H}\mathbf{h}_{bt}\mathbf{h}_{bt}^{{T}}\mathbf{f}_{k1}+\mathbf{w}^{H}\mathbf{h}_{bt}\mathbf{h}_{rt}^{{T}}\mathbf{A}_{a}\mathbf{f}_{k2},
\\&\mathbf{Y}_{1,k}\triangleq(\mathrm{diag}(\mathbf{h}_{rt})\mathbf{H}_{3}\mathbf{f}_{k1})^{T}
\otimes\mathbf{w}^{H}\mathbf{h}_{bt} \notag
\\&\quad\quad\quad+
(\mathbf{h}_{bt}^{{T}}\mathbf{f}_{k1})^{T}\otimes\mathbf{w}^{H}\mathbf{H}_{3}^{T}\mathrm{diag}(\mathbf{h}_{rt})  
\notag
\\&\quad\quad\quad+(\mathbf{h}_{rt}^{{T}}\mathbf{A}_{a}\mathbf{f}_{k2})^{T}
\otimes\mathbf{w}^{H}\mathbf{H}_{3}^{T}\mathrm{diag}(\mathbf{h}_{rt}), 
\\& \mathbf{Z}_{1,k}\triangleq(\mathrm{diag}(\mathbf{h}_{rt}))\mathbf{H}_{3}\mathbf{f}_{k1} )^{T}\otimes\mathbf{w}^{H}\mathbf{H}_{3}^{T}\mathrm{diag}(\mathbf{h}_{rt}). 
\end{align}  \end{subequations}
Based on the above transformations, the first term of (\ref{obj_lag}) can be expressed as 
\begin{equation}\begin{aligned}
f_{\mathrm{obj1},\boldsymbol{\phi}}&=\mathbf{w}^{H}{\mathbf{H}}_{2}\mathbf{F}\mathbf{F}^{H}{\mathbf{H}}_{2}^{H}\mathbf{w}\\&=
\sum_{k=1}^{K}\mathbf{f}_{k}^{H}{\mathbf{H}}_{2}^{H}\mathbf{w}\mathbf{w}^{H}{\mathbf{H}}_{2}\mathbf{f}_{k}
\\&={r}_{1}+2\Re\{\boldsymbol{\phi}^{H}\mathbf{r}_{2}\}+2\Re\{(\boldsymbol{\phi}\otimes\boldsymbol{\phi})^{H}\mathbf{r}_{3}\}\\&\quad
+\boldsymbol{\phi}^{H}\mathbf{R}_{4}\boldsymbol{\phi}+2\Re\{(\boldsymbol{\phi}\otimes\boldsymbol{\phi})^{H}\mathbf{R}_{5}\boldsymbol{\phi}\}
\\& \quad+(\boldsymbol{\phi}\otimes\boldsymbol{\phi})^{H}\mathbf{R}_{6}(\boldsymbol{\phi}\otimes\boldsymbol{\phi}),
\end{aligned}\end{equation}
where we define
\begin{subequations}\begin{align}
&{r}_{1}\triangleq\sum_{k=1}^{K}\mathbf{x}_{1,k}^{H}\mathbf{x}_{1,k}, \quad \mathbf{r}_{2}\triangleq\sum_{k=1}^{K}\mathbf{Y}_{1,k}^{H}\mathbf{x}_{1,k}, \label{r1_r2}\\&\mathbf{r}_{3}\triangleq\sum_{k=1}^{K}\mathbf{Z}_{1,k}^{H}\mathbf{x}_{1,k}, \quad \mathbf{R}_{4}\triangleq\sum_{k=1}^{K}\mathbf{Y}_{1,k}^{H}\mathbf{Y}_{1,k},\label{r3_r4}
\\&\mathbf{R}_{5}\triangleq\sum_{k=1}^{K}\mathbf{Z}_{1,k}^{H}\mathbf{Y}_{1,k},
\quad \mathbf{R}_{6}\triangleq\sum_{k=1}^{K}\mathbf{Z}_{1,k}^{H}\mathbf{Z}_{1,k}.
\label{r5_r6}\end{align}  \end{subequations} 
Based on the same idea, the second term of (\ref{obj_lag}) can be derived as 
\begin{equation}\begin{aligned}
f_{\mathrm{obj2},\boldsymbol{\phi}}&=
\frac{1}{2\rho_{1}}\sum _{k =1}^{K}\sum _{i=1}^{K}|{\mathbf{h}}_{1,k}^{T}\mathbf{f}_{i}-s_{k,i}|^{2}
\\& =\frac{1}{2\rho_{1}}\sum _{k =1}^{K}\sum _{i=1}^{K}({\mathbf{h}}_{1,k}^{T}\mathbf{f}_{i}-s_{k,i})({\mathbf{h}}_{1,k}^{T}\mathbf{f}_{i}-s_{k,i})^{H}
\\& =\frac{1}{2\rho_{1}}\sum _{k =1}^{K}\sum _{i=1}^{K}(\tilde{s}_{k,i}+\boldsymbol{\phi}^{T}\mathrm{diag}(\mathbf{h}_{ru,k})\mathbf{H}_{3}\mathbf{f}_{i1})(\tilde{s}_{k,i}^{H}\\&\quad+\mathbf{f}_{i1}^{H}\mathbf{H}_{3}^{H}\mathrm{diag}(\mathbf{h}_{ru,k}^{\ast})\boldsymbol{\phi}^{\ast})
\\&=\frac{1}{2\rho_{1}}(r_{7}+2\Re\{\mathbf{r}_{8}^{H}\boldsymbol{\phi}\}+\boldsymbol{\phi}^{T}\mathbf{R}_{9}\boldsymbol{\phi}^{\ast}),
\end{aligned}\end{equation}
where we define
\begin{subequations}\begin{align}
&\tilde{s}_{k,i}\triangleq{\mathbf{h}}_{bu,k}^{T}\mathbf{f}_{i1}+\mathbf{h}_{ru,k}^{T}\mathbf{A}_{a}\mathbf{f}_{i2}-s_{k,i},
\\& r_{7}\triangleq\sum _{k =1}^{K}\sum _{i=1}^{K}\tilde{s}_{k,i}\tilde{s}_{k,i}^{H}, \label{r7}\\& \mathbf{r}_{8}\triangleq\sum _{k =1}^{K}\sum _{i=1}^{K}\mathrm{diag}(\mathbf{h}_{ru,k})\mathbf{H}_{3}\mathbf{f}_{i1}\tilde{s}_{k,i}^{H}, \label{r8}\\& \mathbf{R}_{9}\triangleq\sum _{k =1}^{K}\sum _{i=1}^{K}\mathrm{diag}(\mathbf{h}_{ru,k})\mathbf{H}_{3}\mathbf{f}_{i1}\mathbf{f}_{i1}^{H}\mathbf{H}_{3}^{H}\mathrm{diag}(\mathbf{h}_{ru,k}^{\ast})\label{r9}.
\end{align}  \end{subequations} 
Omitting the irrelevant terms, the overall objective function w.r.t. $\boldsymbol{\phi}$ can be written as
\begin{equation}\begin{aligned}
&f_{\mathrm{obj},\boldsymbol{\phi}}\\&=
f_{\mathrm{obj1},\boldsymbol{\phi}}+f_{\mathrm{obj2},\boldsymbol{\phi}}\\&=
-({r}_{1}+2\Re\{\mathbf{r}_{2}^{H}\boldsymbol{\phi}\}+2\Re\{\mathbf{r}_{3}^{H}(\boldsymbol{\phi}\otimes\boldsymbol{\phi})\}
+\boldsymbol{\phi}^{H}\mathbf{R}_{4}\boldsymbol{\phi}\\& \quad +2\Re\{(\boldsymbol{\phi}\otimes\boldsymbol{\phi})^{H}\mathbf{R}_{5}\boldsymbol{\phi}\}
+(\boldsymbol{\phi}\otimes\boldsymbol{\phi})^{H}\mathbf{R}_{6}(\boldsymbol{\phi}\otimes\boldsymbol{\phi}))\\& \quad +\frac{1}{2\rho_{1}}(r_{7}+2\Re\{\mathbf{r}_{8}^{H}\boldsymbol{\phi}\}+\boldsymbol{\phi}^{T}\mathbf{R}_{9}\boldsymbol{\phi}^{\ast}).
\end{aligned}\end{equation}
Following the same process, the first two parts of the objective function w.r.t $\mathbf{a}$ and $\mathbf{a}_{a}$ are expressed as
\begin{equation}\begin{aligned}
&f_{\mathrm{obj1},\mathbf{a}}+f_{\mathrm{obj2},\mathbf{a}}\\&=-(r_{10}+2\Re\{\mathbf{r}_{11}^{H}\mathbf{a}\}
+2\Re\{\mathbf{r}_{12}^{H}(\mathbf{a}\!\otimes\!\mathbf{a})\}
+\mathbf{a}^{T}\mathbf{R}_{13}\mathbf{a}
\\& \quad
+2\Re\{(\mathbf{a}\!\otimes\!\mathbf{a})^{T}\mathbf{R}_{14}\mathbf{a}\}
+(\mathbf{a}\!\otimes\!\mathbf{a})^{T}\mathbf{R}_{15}(\mathbf{a}\!\otimes\!\mathbf{a}))
\\&\quad+\frac{1}{2\rho_{1}}(r_{16}+2\Re\{\mathbf{r}_{17}^{H}\mathbf{a}\}+\mathbf{a}^{T}\mathbf{R}_{18}\mathbf{a}),
\end{aligned}\end{equation}
and 
\begin{equation}\begin{aligned}
&f_{\mathrm{obj1},\mathbf{a}_{a}}+f_{\mathrm{obj2},\mathbf{a}_{a}}\\&=-(r_{21}+2\Re\{{\mathbf{r}}_{22}^{H}\mathbf{a}_{a}+\mathbf{a}_{a}^{T}{\mathbf{R}}_{23}\mathbf{a}_{a}\})+\frac{1}{2\rho_{1}}(r_{24}\\&\quad+2\Re\{\mathbf{r}_{25}^{H}\mathbf{a}_{a}\}+\mathbf{a}_{a}^{T}\mathbf{R}_{26}\mathbf{a}_{a}),
\end{aligned}\end{equation}
respectively. $r_{10}$, $\mathbf{r}_{11}$, $\mathbf{r}_{12}$, $\mathbf{R}_{13},\cdots,\mathbf{R}_{15}$, $r_{16}$, $\mathbf{r}_{17}$, $\mathbf{R}_{18}$, $r_{21}$, $\mathbf{r}_{22}$, $\mathbf{R}_{23}$, ${r}_{24}$, $\mathbf{r}_{25}$ and $\mathbf{R}_{26}$ 
have the similar forms with ${r}_{1}$, $\mathbf{r}_{2}$, $\mathbf{r}_{3}$, $\mathbf{R}_{4},\cdots, \mathbf{R}_{6}$, ${r}_{7}$, $\mathbf{r}_{8}$, $\mathbf{R}_{9}$, so we omit the details of them due to the space limitations. 
For the third part of the objective function w.r.t $\mathbf{a}$ and $\mathbf{a}_{a}$, we have
\begin{subequations}\begin{align}
f_{\mathrm{obj3},\mathbf{a}}&=
\frac{1}{2\rho_{2}}\lVert  \mathbf{A}-\mathbf{A}_{a}\mathbf{A}_{a}^{T}\rVert_{F}^{2}
\notag\\ & =\frac{1}{2\rho_{2}}\mathrm{tr}[(\mathbf{A}^{T}-\mathbf{A}_{a}\mathbf{A}_{a}^{T})(\mathbf{A}-\mathbf{A}_{a}\mathbf{A}_{a}^{T})]
\notag\\ & =\frac{1}{2\rho_{2}}\mathrm{tr}[\mathbf{A}(\mathbf{I}-2\mathbf{A}_{a}\mathbf{A}_{a}^{T})+\mathbf{A}_{a}\mathbf{A}_{a}^{T}]
\notag\\ & =\frac{1}{2\rho_{2}}(\mathbf{r}_{19}^{T}\mathbf{a}+r_{20}),
\\f_{\mathrm{obj3},\mathbf{a}_{a}}&=
\frac{1}{2\rho_{2}}\lVert  \mathbf{A}-\mathbf{A}_{a}\mathbf{A}_{a}^{T}\rVert_{F}^{2}
\notag\\ & =\frac{1}{2\rho_{2}}\mathrm{tr}[(\mathbf{A}^{T}-\mathbf{A}_{a}\mathbf{A}_{a}^{T})(\mathbf{A}-\mathbf{A}_{a}\mathbf{A}_{a}^{T})]
\notag\\ & =\frac{1}{2\rho_{2}}[\mathrm{tr}(\mathbf{A}+\mathbf{A}_{a}\mathbf{A}_{a}^{T})-2\mathrm{tr}(\mathbf{A}\mathbf{A}_{a}\mathbf{A}_{a}^{T})]
\notag\\ & =  \frac{1}{\rho_{2}}[a-\mathbf{a}_{a}^{T}(\mathbf{I}_{a}\otimes\mathbf{A})\mathbf{a}_{a}],
\end{align}\end{subequations}
where we define
\begin{equation}\begin{aligned}
&\mathbf{r}_{19}\triangleq\mathrm{diag}(\mathbf{I}-2\mathbf{A}_{a}\mathbf{A}_{a}^{T}), 
\quad r_{20}\triangleq\mathrm{tr}(\mathbf{A}_{a}\mathbf{A}_{a}^{T}). 
\label{r19_r20}\end{aligned}\end{equation}
Then, the explicit form of the objective functions w.r.t.  $\mathbf{a}$ and $\mathbf{a}_{a}$ can be obtained as
\begin{align}
f_{\mathrm{obj},\mathbf{a}}\notag&=f_{\mathrm{obj1},\mathbf{a}}+f_{\mathrm{obj2},\mathbf{a}}+f_{\mathrm{obj3},\mathbf{a}}
\notag\\&=-(r_{10}+2\Re\{\mathbf{r}_{11}^{H}\mathbf{a}\}
+2\Re\{\mathbf{r}_{12}^{H}(\mathbf{a}\!\otimes\!\mathbf{a})\}
+\mathbf{a}^{T}\mathbf{R}_{13}\mathbf{a}
\notag\\& \quad
+2\Re\{(\mathbf{a}\!\otimes\!\mathbf{a})^{T}\mathbf{R}_{14}\mathbf{a}\}
+(\mathbf{a}\!\otimes\!\mathbf{a})^{T}\mathbf{R}_{15}(\mathbf{a}\!\otimes\!\mathbf{a}))
\!+\!\frac{1}{2\rho_{1}}\notag\\&\quad \quad(r_{16}+2\Re\{\mathbf{r}_{17}^{H}\mathbf{a}\}+\mathbf{a}^{T}\mathbf{R}_{18}\mathbf{a})+\frac{1}{2\rho_{2}}(\mathbf{r}_{19}^{T}\mathbf{a}+r_{20}),
\end{align}
and 
\begin{equation}\begin{aligned}
f_{\mathrm{obj},\mathbf{a}_{a}}&=f_{\mathrm{obj1},\mathbf{a}_{a}}+f_{\mathrm{obj2},\mathbf{a}_{a}}+f_{\mathrm{obj3},\mathbf{a}_{a}}
\\&=-(r_{21}+2\Re\{{\mathbf{r}}_{22}^{H}\mathbf{a}_{a}+\mathbf{a}_{a}^{T}{\mathbf{R}}_{23}\mathbf{a}_{a}\})+\frac{1}{2\rho_{1}}(r_{24}+
\\&\quad \quad 2\Re\{\mathbf{r}_{25}^{H}\mathbf{a}_{a}\}+\mathbf{a}_{a}^{T}\mathbf{R}_{26}\mathbf{a}_{a})\!+\!\frac{1}{\rho_{2}}(a-\mathbf{a}_{a}^{T}(\mathbf{I}_{a}\otimes\mathbf{A})\mathbf{a}_{a}),
\end{aligned}\end{equation}
respectively. Thus, we complete the proof of Lemma 1.


\end{appendices}

 \vspace{-2mm}

\begin{thebibliography}{10}
\providecommand{\url}[1]{#1}
\csname url@samestyle\endcsname
\providecommand{\newblock}{\relax}
\providecommand{\bibinfo}[2]{#2}
\providecommand{\BIBentrySTDinterwordspacing}{\spaceskip=0pt\relax}
\providecommand{\BIBentryALTinterwordstretchfactor}{4}
\providecommand{\BIBentryALTinterwordspacing}{\spaceskip=\fontdimen2\font plus
\BIBentryALTinterwordstretchfactor\fontdimen3\font minus \fontdimen4\font\relax}
\providecommand{\BIBforeignlanguage}[2]{{%
\expandafter\ifx\csname l@#1\endcsname\relax
\typeout{** WARNING: IEEEtran.bst: No hyphenation pattern has been}%
\typeout{** loaded for the language `#1'. Using the pattern for}%
\typeout{** the default language instead.}%
\else
\language=\csname l@#1\endcsname
\fi
#2}}
\providecommand{\BIBdecl}{\relax}
\BIBdecl

\bibitem{1}
\BIBentryALTinterwordspacing
P.Brown, ``75.4 billion devices connected to the internet of things by 2025,'' \emph{Electronics}, vol. 360, 2016. [Online]. Available: \url{https://electronics360.globalspec.com/article/6551/75-4-billion-devices-connected-to-the-internet-of-things-by-2025}
\BIBentrySTDinterwordspacing

\bibitem{zhang2021enabling}
J.~A. Zhang, M.~L. Rahman, K.~Wu, X.~Huang, Y.~J. Guo, S.~Chen, and J.~Yuan, ``Enabling joint communication and radar sensing in mobile networks—a survey,'' \emph{IEEE Commun. Surv. Tutor.}, vol.~24, no.~1, pp. 306--345, 2021.

\bibitem{liu2020joint}
F.~Liu, C.~Masouros, A.~P. Petropulu, H.~Griffiths, and L.~Hanzo, ``Joint radar and communication design: Applications, state-of-the-art, and the road ahead,'' \emph{{IEEE} Trans. Commun.}, vol.~68, no.~6, pp. 3834--3862, 2020.

\bibitem{liu2022survey}
A.~Liu, Z.~Huang, M.~Li, Y.~Wan, W.~Li, T.~X. Han, C.~Liu, R.~Du, D.~K.~P. Tan, J.~Lu \emph{et~al.}, ``A survey on fundamental limits of integrated sensing and communication,'' \emph{IEEE Commun. Surv. Tutor.}, vol.~24, no.~2, pp. 994--1034, 2022.

\bibitem{liu2022integrated}
F.~Liu, Y.~Cui, C.~Masouros, J.~Xu, T.~X. Han, Y.~C. Eldar, and S.~Buzzi, ``Integrated sensing and communications: Towards dual-functional wireless networks for {6G} and beyond,'' \emph{{IEEE} J. Sel. Areas Commun.}, 2022.

\bibitem{yuan2021integrated}
W.~Yuan, Z.~Wei, S.~Li, J.~Yuan, and D.~W.~K. Ng, ``Integrated sensing and communication-assisted orthogonal time frequency space transmission for vehicular networks,'' \emph{{IEEE} J. Sel. Topics Signal Process.}, vol.~15, no.~6, pp. 1515--1528, 2021.

\bibitem{8892631}
Z.~Cheng, B.~Liao, S.~Shi, Z.~He, and J.~Li, ``Co-design for overlaid {MIMO} radar and downlink {MISO} communication systems via {C}ramér–{R}ao bound minimization,'' \emph{{IEEE} Trans. Signal Process.}, vol.~67, no.~24, pp. 6227--6240, 2019.

\bibitem{wang2023stars}
Z.~Wang, X.~Mu, and Y.~Liu, ``{STARS} enabled integrated sensing and communications,'' \emph{{IEEE} Trans. Wireless Commun.}, 2023.

\bibitem{10056867}
J.~Wang, S.~Gong, Q.~Wu, and S.~Ma, ``{RIS}-aided {MIMO} systems with hardware impairments: Robust beamforming design and analysis,'' \emph{{IEEE} Trans. Wireless Commun.}, pp. 1--1, 2023.

\bibitem{zhang2023double}
P.~Zhang, S.~Gong, and S.~Ma, ``Double-{RIS} aided multi-user {MIMO} communications: Common reflection pattern and joint beamforming design,'' \emph{{IEEE} Trans. Veh. Technol.}, pp. 1--6, 2023.

\bibitem{9429987}
S.~Gong, C.~Xing, X.~Zhao, S.~Ma, and J.~An, ``Unified {IRS}-aided {MIMO} transceiver designs via majorization theory,'' \emph{{IEEE} Trans. Signal Process.}, vol.~69, pp. 3016--3032, 2021.

\bibitem{9364358}
Z.-M. Jiang, M.~Rihan, P.~Zhang, L.~Huang, Q.~Deng, J.~Zhang, and E.~M. Mohamed, ``Intelligent reflecting surface aided dual-function radar and communication system,'' \emph{IEEE Syst. J.}, vol.~16, no.~1, pp. 475--486, 2022.

\bibitem{luo2023ris}
H.~Luo, R.~Liu, M.~Li, and Q.~Liu, ``{RIS}-aided integrated sensing and communication: Joint beamforming and reflection design,'' \emph{{IEEE} Trans. Veh. Technol.}, 2023.

\bibitem{xing2022passive}
Z.~Xing, R.~Wang, and X.~Yuan, ``Passive beamforming design for reconfigurable intelligent surface enabled integrated sensing and communication,'' \emph{arXiv preprint arXiv:2206.00525}, 2022.

\bibitem{9844707}
H.~Zhang, ``Joint waveform and phase shift design for {RIS}-assisted integrated sensing and communication based on mutual information,'' \emph{{IEEE} Commun. Lett.}, vol.~26, no.~10, pp. 2317--2321, 2022.

\bibitem{liu2022joint}
R.~Liu, M.~Li, Y.~Liu, Q.~Wu, and Q.~Liu, ``Joint transmit waveform and passive beamforming design for {RIS}-aided {DFRC} systems,'' \emph{{IEEE} J. Sel. Topics Signal Process.}, vol.~16, no.~5, pp. 995--1010, 2022.

\bibitem{xu2023joint}
Y.~Xu, Y.~Li, J.~A. Zhang, and M.~Di~Renzo, ``Joint beamforming for {RIS}-assisted integrated sensing and communication systems,'' \emph{arXiv preprint arXiv:2303.01771}, 2023.

\bibitem{he2022ris}
Y.~He, Y.~Cai, H.~Mao, and G.~Yu, ``{RIS}-assisted communication radar coexistence: Joint beamforming design and analysis,'' \emph{{IEEE} J. Sel. Areas Commun.}, vol.~40, no.~7, pp. 2131--2145, 2022.

\bibitem{9591331}
X.~Wang, Z.~Fei, J.~Huang, and H.~Yu, ``Joint waveform and discrete phase shift design for {RIS}-assisted integrated sensing and communication system under {Cramer-Rao} bound constraint,'' \emph{{IEEE} Trans. Veh. Technol.}, vol.~71, no.~1, pp. 1004--1009, 2022.

\bibitem{10138058}
X.~Song, J.~Xu, F.~Liu, T.~X. Han, and Y.~C. Eldar, ``Intelligent reflecting surface enabled sensing: Cramér-rao bound optimization,'' \emph{{IEEE} Trans. Signal Process.}, vol.~71, pp. 2011--2026, 2023.

\bibitem{9998527}
Z.~Zhang, L.~Dai, X.~Chen, C.~Liu, F.~Yang, R.~Schober, and H.~V. Poor, ``Active {RIS} vs. passive {RIS}: Which will prevail in {6G}?'' \emph{{IEEE} Trans. Commun.}, vol.~71, no.~3, pp. 1707--1725, 2023.

\bibitem{long2021active}
R.~Long, Y.-C. Liang, Y.~Pei, and E.~G. Larsson, ``Active reconfigurable intelligent surface-aided wireless communications,'' \emph{{IEEE} Trans. Wireless Commun.}, vol.~20, no.~8, pp. 4962--4975, 2021.

\bibitem{sankar2022beamforming}
R.~P. Sankar and S.~P. Chepuri, ``Beamforming in hybrid {RIS} assisted integrated sensing and communication systems,'' in \emph{2022 30th European Signal Processing Conference (EUSIPCO)}.\hskip 1em plus 0.5em minus 0.4em\relax IEEE, 2022, pp. 1082--1086.

\bibitem{salem2022active}
A.~A. Salem, M.~H. Ismail, and A.~S. Ibrahim, ``Active reconfigurable intelligent surface-assisted {MISO} integrated sensing and communication systems for secure operation,'' \emph{{IEEE} Trans. Veh. Technol.}, 2022.

\bibitem{yu2023active}
Z.~Yu, G.~Zhou, H.~Ren, C.~Pan, B.~Wang, M.~Dong, and J.~Wang, ``Active {RIS} aided integrated sensing and communication systems,'' \emph{arXiv preprint arXiv:2302.08934}, 2023.

\bibitem{9979782}
Y.~Zhang, J.~Chen, C.~Zhong, H.~Peng, and W.~Lu, ``Active {IRS}-assisted integrated sensing and communication in {C-RAN},'' \emph{{IEEE} Wireless Commun. Lett.}, vol.~12, no.~3, pp. 411--415, 2023.

\bibitem{10184278}
W.~Hao, H.~Shi, G.~Sun, and C.~Huang, ``Joint beamforming design for active {RIS}-aided {THz} {ISAC} systems with delay alignment modulation,'' \emph{{IEEE} Commun. Lett.}, vol.~12, no.~10, pp. 1816--1820, 2023.

\bibitem{ma2023reconfigurable}
C.~Ma, X.~Yang, J.~Wang, G.~Yang, and S.~Ma, ``Reconfigurable distributed antennas and reflecting surface ({RDARS}): A new architecture for wireless communications,'' \emph{arXiv preprint arXiv:2303.06950}, 2023.

\bibitem{10197455}
X.~Qian, X.~Hu, C.~Liu, M.~Peng, and C.~Zhong, ``Sensing-based beamforming design for joint performance enhancement of {RIS}-aided {ISAC} systems,'' \emph{{IEEE} Trans. Commun.}, vol.~71, no.~11, pp. 6529--6545, 2023.

\bibitem{9854847}
W.~Yang, M.~Li, and Q.~Liu, ``A novel anchor-assisted channel estimation for {RIS}-aided multiuser communication systems,'' \emph{{IEEE} Commun. Lett.}, vol.~26, no.~11, pp. 2740--2744, 2022.

\bibitem{9366805}
L.~Wei, C.~Huang, G.~C. Alexandropoulos, C.~Yuen, Z.~Zhang, and M.~Debbah, ``Channel estimation for {RIS}-empowered multi-user {MISO} wireless communications,'' \emph{{IEEE} Trans. Commun.}, vol.~69, no.~6, pp. 4144--4157, 2021.

\bibitem{zhou2022channel}
G.~Zhou, C.~Pan, H.~Ren, P.~Popovski, and A.~L. Swindlehurst, ``Channel estimation for {RIS}-aided multiuser millimeter-wave systems,'' \emph{{IEEE} Trans. Signal Process.}, vol.~70, pp. 1478--1492, 2022.

\bibitem{sun2016majorization}
Y.~Sun, P.~Babu, and D.~P. Palomar, ``Majorization-minimization algorithms in signal processing, communications, and machine learning,'' \emph{{IEEE} Trans. Signal Process.}, vol.~65, no.~3, pp. 794--816, 2016.

\bibitem{CVX}
\BIBentryALTinterwordspacing
M.~Grant and S.~Boyd, ``{CVX}: {MATLAB} software for disciplined convex programming,'' 2016. [Online]. Available: \url{http://cvxr.com/cvx/}
\BIBentrySTDinterwordspacing

\bibitem{9133435}
Q.~Wu and R.~Zhang, ``Joint active and passive beamforming optimization for intelligent reflecting surface assisted {SWIPT} under {QoS} constraints,'' \emph{{IEEE} J. Sel. Areas Commun.}, vol.~38, no.~8, pp. 1735--1748, 2020.

\bibitem{boyd_vandenberghe_2004}
S.~Boyd and L.~Vandenberghe, \emph{Convex Optimization}.\hskip 1em plus 0.5em minus 0.4em\relax Cambridge University Press, 2004.

\bibitem{9913311}
M.~Hua, Q.~Wu, C.~He, S.~Ma, and W.~Chen, ``Joint active and passive beamforming design for {IRS}-aided radar-communication,'' \emph{{IEEE} Trans. Wireless Commun.}, vol.~22, no.~4, pp. 2278--2294, 2023.

\bibitem{shi2011iteratively}
Q.~Shi, M.~Razaviyayn, Z.-Q. Luo, and C.~He, ``An iteratively weighted {MMSE} approach to distributed sum-utility maximization for a {MIMO} interfering broadcast channel,'' \emph{{IEEE} Trans. Signal Process.}, vol.~59, no.~9, pp. 4331--4340, 2011.

\bibitem{7558213}
Q.~Shi, M.~Hong, X.~Gao, E.~Song, Y.~Cai, and W.~Xu, ``Joint source-relay design for full-duplex {MIMO} {AF} relay systems,'' \emph{{IEEE} Trans. Signal Process.}, vol.~64, no.~23, pp. 6118--6131, 2016.

\end{thebibliography}

\end{document}